\documentclass[11pt]{article}

\usepackage{paralist}
\usepackage{color}
\usepackage{bm}

\usepackage{soul}
\usepackage[T1]{fontenc}%

\usepackage[cm]{fullpage}
\usepackage{amsfonts}
\usepackage{amssymb}
\usepackage{multirow}
\usepackage{amsmath}
\usepackage{constants}
\usepackage{amsthm}

\usepackage{footnote}
\makesavenoteenv{tabular}
\makesavenoteenv{table}
\makeatletter
\newtheorem*{rep@theorem}{\rep@title}
\newcommand{\newreptheorem}[2]{%
\newenvironment{rep#1}[1]{%
 \def\rep@title{#2 \ref{##1}}%
 \begin{rep@theorem}}%
 {\end{rep@theorem}}}
\makeatother

\usepackage{ifpdf}
\newcommand{\mydriver}{hypertex}
\ifpdf
 \renewcommand{\mydriver}{pdftex}
\fi
\usepackage[breaklinks,\mydriver]{hyperref}

\usepackage{mathtools}
\usepackage[margin=1in]{geometry}

\usepackage{enumitem}
\setlist[enumerate]{itemsep=3pt,topsep=3pt}

\newcommand{\Thr}{\ensuremath{T}}

\newcommand{\E}{\textrm{E}}

\newcommand{\poly}{\textrm{poly}}

\newcommand{\nis}{\ensuremath{\mathrm{nis}}}
\newcommand{\nscc}{\ensuremath{\mathrm{nscc}}}
\newcommand{\cnt}{\ensuremath{\mathrm{cnt}}}

\newcommand{\Insert}{\ensuremath{\mathbf{Insert}}}
\newcommand{\Delete}{\ensuremath{\mathbf{Delete}}}
\newcommand{\Query}{\ensuremath{\mathbf{Query}}}

\theoremstyle{plain}
\newtheorem{theorem}{Theorem}[section]%

\newtheorem{lemma}[theorem]{Lemma}
\newtheorem{corollary}[theorem]{Corollary}
\newtheorem{claim}[theorem]{Claim}

\newreptheorem{theorem}{Theorem}
\newreptheorem{lemma}{Lemma}

\theoremstyle{definition}

\newcommand{\ncc}{\ensuremath{\mathrm{ncc}}}

\newcommand{\err}{{\textsc{Err}}}

\newcommand{\NewL}[1]{L_{#1,\textrm{new}}}
\newcommand{\OldL}[1]{L_{#1,\textrm{old}}}
\newcommand{\RNewL}[1]{R_{#1,L_\textrm{new}}}
\newcommand{\ROldL}[1]{R_{#1,L_\textrm{old}}}

\newcommand{\GoodL}[1]{L_{#1}^{\textrm{<}}}
\newcommand{\GoodNewL}[1]{L_{#1,\textrm{new}}^{\textrm{<}}}
\newcommand{\GoodOldL}[1]{L_{#1,\textrm{old}}^{\textrm{<}}}

\newcommand{\medianN}[1]{m_{#1,{\textrm{new}}}}
\newcommand{\medianO}[1]{m_{#1,L_{\textrm{old}}}}

\newcommand{\NewN}[1]{N_{\textrm{new}}(#1)}
\newcommand{\OldN}[1]{N_{\textrm{old}}(#1)}

\newcommand{\FirstLower}{{\mathcal{F}}}

\newcommand{\ColorM}[2]{\mathcal{M}_{#1}(#2)}
\newcommand{\ColorU}[2]{\ensuremath{\mathcal{U}_{#1}(#2)}}
\newcommand{\ColorB}[1]{\mathcal{B}_{#1}}
\newcommand{\ColorH}[1]{\ensuremath{\mathcal{C}_{#1}(H)}}
\newcommand{\ColorLL}[1]{\ensuremath{\mathcal{C}_{#1}(L)}}
\newcommand{\ColorL}[1]{\ensuremath{\mathcal{C}_{#1}(\overline H)}}

\newcommand{\junk}[1]{{}}
\newcommand{\pnew}[1]{#1}

\newcommand{\Ra}[1]{\mathcal{R}{(#1)}}

\providecommand{\abs}[1]{\lvert#1\rvert}

\usepackage{enumitem}

\title{\Large Constant-Time Dynamic $(\Delta+1)$-Coloring and Weight Approximation for Minimum Spanning Forest: Dynamic Algorithms Meet Property Testing}

\author{
	Monika Henzinger\footnote{University of Vienna, Faculty of Computer Science, Vienna, Austria. E-mail: \texttt{monika.henzinger@univie.ac.at}. The research leading to these results has received funding from the European
		Research Council under the European Union's Seventh Framework Programme
		(FP/2007-2013) / ERC Grant Agreement no. 340506.}
	\and 
	Pan Peng\footnote{Department of Computer Science, University of Sheffield, Sheffield, UK. Email: \texttt{p.peng@sheffield.ac.uk}.}
}
\date{}

\begin{document}
	\begin{titlepage}
		\maketitle
		\thispagestyle{empty}

		\begin{abstract}
With few exceptions (namely, algorithms for maximal matching, $2$-approximate vertex cover, and certain
constant-stretch spanners), all known fully dynamic algorithms  in general graphs require (amortized) $\Omega(\log n)$ update/query time. 
Showing for the first time that techniques from property testing can lead to constant-time fully dynamic graph algorithms we prove the following results:

(1) We give a fully dynamic (Las-Vegas style) algorithm with {\em constant expected amortized}
time per update that maintains a proper $(\Delta+1)$-vertex coloring of a graph with maximum degree at most $\Delta$. This improves upon the previous $O(\log \Delta)$-time algorithm by Bhattacharya et al. (SODA 2018). 
We show that our result does not only have optimal running time, but is also optimal in the sense that  already deciding whether a $\Delta$-coloring  exists 
	in a dynamically changing graph with maximum degree at most $\Delta$ takes $\Omega(\log n)$ time per operation.

(2) We give two fully dynamic algorithms that maintain a $(1+\varepsilon)$-approximation of the weight $M$ of the minimum spanning forest of a graph $G$ with edges weights in $[1,W]$. Our \emph{deterministic} algorithm takes $O({W^2 \log W}/{\varepsilon^3})$ {\em worst-case} time, which is constant if both $W$ and $\varepsilon$ are constant.
This is somewhat surprising as a lower bound by Patrascu and Demaine (SIAM J. Comput. 2006) shows that it takes $\Omega(\log n)$ time per operation to maintain the {\em exact}
weight of the MSF that holds even for $W=1$. 
Our randomized (Monte-Carlo style) algorithm works with high probability and runs in \emph{worst-case} $O(\frac{1}{\varepsilon^4}\log^2(\frac{1}{\varepsilon}))$ time if $W=  O({(m^*)^{1/3}}/{\log^3 n})$, where $m^*$ is the minimum number of edges in the graph throughout all the updates. It works even against an adaptive adversary.

 \end{abstract}
\end{titlepage}
\section{Introduction}
A (fully) dynamic graph algorithm is a data structure that provides information about a graph property while the graph is being modified by {\em edge updates} such as edge insertions or deletions. When designing a dynamic graph algorithm the goal is to minimize the time per update or query operation.
The  lower bounds of Patrascu and Demaine~\cite{PatrascuD06} 
showed that in the cell-probe model many fundamental graph properties, such as asking whether the graph is connected, require $\Omega(\log n)$ time per operation, where $n$ is the number of nodes in the graph. Their lower bound technique also gives logarithmic time lower bounds for further dynamic problems such as higher types of connectivity, planarity and  bipartiteness testing, and minimum spanning forest, and it is an open research question for which other dynamic graph problems non-constant time lower bounds exist.

Furthermore, there are only very few graph problems for which it is known that no such lower bounds can exist. These are the following problems, which all have constant-time algorithms:
maintaining (a) a maximal matching (randomized)~\cite{Sol16:matching}, (b) a $(2+\varepsilon)$-approximate vertex cover (deterministic)~\cite{BhattacharyaK2019}, and (c)
 a $(2k-1)$-stretch spanner of size $O(n^{1+\frac{1}{k}}\log^2 n)$ for \emph{constant} $k$ (randomized)~\cite{BKS12:spanner}. 
All these are \emph{amortized} time bounds and all these algorithms maintain a sophisticated hierarchical graph decomposition, which makes them rather impractical.

Techniques from distributed, streaming, and online algorithms have been used  in the past to design efficient dynamic graph algorithms (see also the {\em Related Work Section} in Section~\ref{sec:related_work}). 
However, we are not aware of any dynamic graph algorithm in general graphs that exploits techniques from {\em sublinear-time algorithms  and property testing} 
and one goal of this paper is to push forward the 
study of the promising connection between these two fields.
Intuitively, in both fields, dynamic graph algorithms and property testing, we try to find out information about a graph using as little (time) resources as possible and, thus, we want to probe only very few ``places'' in the graph. As we show this intuition can indeed be exploited to achieve new constant-time dynamic graph algorithms.

\subsection{Our Contributions}
Our first and main contribution is a new randomized\footnote{For randomized dynamic algorithms, we assume the (almost) standard \emph{oblivious adversary}: the adversary who fixes the sequence of edge insertions
and deletions is \emph{oblivious} to the randomness in our algorithm.} dynamic algorithm for vertex coloring. Given a graph let $\Delta$ be an upper bound on the maximum degree in the graph. 
A {\em proper coloring} assigns to each vertex an integer value, called {\em color}, such that the endpoints of every edge have a different color. A $(\Delta+1)$-vertex coloring is a proper coloring that uses only colors from the range $[1, \dots, \Delta+1]$. %
{It was known that a proper $(\Delta+1)$-vertex coloring in a (static) graph with maximum degree at most $\Delta$ always exists and can be found in linear time by a simple greedy algorithm. Dynamically maintaining a proper $(\Delta+1)$-vertex coloring was investigated only very recently by Bhattacharya et al.~\cite{BCHN18:coloring}, who observed a trivial algorithm with $O(\Delta)$ worst-case update time (which simply scans the whole neighborhood once an edge is inserted between two nodes of the same color), and gave a randomized algorithm\footnote{In \cite{BCHN18:coloring}, the authors also provided a deterministic algorithm that maintains a $(\Delta+o(\Delta))$-vertex coloring with $O(\poly\log \Delta)$ amortized update time.} for this problem with $O(\log \Delta)$ expected amortized time per operation. Note that if $\Delta$ is polynomial in $n$, their algorithm takes time
$O(\log n)$. In this paper, we improve upon their algorithm and prove the following result.} We call a dynamic graph \emph{$\Delta$-bounded} if throughout the updates, the graph has maximum degree at most $\Delta$. 
\begin{theorem}\label{thm:coloring}
	There exists a fully dynamic algorithm for maintaining a proper $(\Delta+1)$-vertex coloring for a $\Delta$-bounded graph with $O(1)$ expected amortized update time.%
\end{theorem}

Apart from having optimal running time, our result is also optimal in the sense that deciding whether a proper coloring with only $\Delta$ colors exists 
	in a dynamically changing graph (with maximum degree at most $\Delta$) takes at least  $\Omega(\log n)$ time per operation, as we show in Appendix~\ref{sec:lb_delta_coloring}.

Our second contribution is two new dynamic algorithms for maintaining the approximate weight of a {\em minimum spanning forest} ({\em MSF}). 	
Given an edge-weighted graph $G$ an {MSF} is a subgraph of $G$ that forms a spanning forest 
and has minimum weight among all spanning forests of $G$. The weight of an MSF is the sum of the edge weights of the MSF.  A $(1+\varepsilon)$-approximation of the weight $M$ of an MSF is a value $M'$ such that $(1-\varepsilon) \cdot M \le M' \le (1 +\varepsilon) \cdot M$.
We show that a $(1+\varepsilon)$-approximation of the weight of an MSF can be maintained deterministically in a graph with edge weights in the interval $[1, W]$ in worst-case time $O(W^2 \log W/\varepsilon^3)$ per operation.
For constant $W$ and $\varepsilon$ this is a \emph{constant deterministic worst-case} time bound.
This result is somewhat surprising, as the lower bound of $\Omega(\log n)$ by~\cite{PatrascuD06}  applies for maintaining the {\em exact} weight of an MSF even for $W=1$.

We also give a randomized dynamic algorithm that, with high probability, maintains a $(1+\varepsilon)$-approximation of the weight $M$ of the MSF of a dynamic graph in
	worst-case time $O({\log^2(1/{\varepsilon})}/{\varepsilon^4})$ if $W=  O(\max\{1, {(m^*)^{1/3}}/{\log^3 n}\})$, where $m^*$ is the minimum number of edges in the graph throughout all the updates. This is useful if the graph always has $\omega(\poly(\log n))$ edges. 
	Interestingly, our algorithm works against an adaptive adversary, which is an adversary that sees the answers to all query operations before deciding which edge to update next.

Our algorithms (with constant $\varepsilon$ and small $W$) are much faster than the best known 
algorithms for maintaining the {\em exact} weight of an MSF: They 
assume only that $W$ is polynomial in $n$,  but take $O(\log^4 n/\log\log n)$ expected amortized time per operation~\cite{HRW15:MST} (that improves upon \cite{HLT01:connectivity}) 
and $O(n^{o(1)})$ expected worst-case time per operation~\cite{NSW17:MST}. 
{Furthermore, combining the techniques of~\cite{HenzingerK01} and of~\cite{KapronKM13}, one can also maintain a spanning forest whose weight is a $(1+\varepsilon)$-approximate of the weight of an MSF with $O(\log^4n/\varepsilon)$ worst-case update time.}

As we recently learnt, Bhattacharya et al.~\cite{BComm} achieved a randomized dynamic algorithm for $(\Delta+1)$-vertex coloring with constant amortized update time independently.

\subsection{Our Techniques}
 We use the following  techniques in our algorithms.

(1)  {\em $(\Delta+1)$-vertex coloring.} 
We first give a brief overview of the algorithm in~\cite{BCHN18:coloring} that maintains a proper $(\Delta+1)$-vertex coloring for a dynamic graph with maximum degree at most $\Delta$. First note that for edge deletion, there is no need to change an existing proper coloring, denoted by $\chi$. For an edge insertion $(u,v)$, if it does not cause a conflict, i.e., $\chi(u)\neq \chi(v)$, then the coloring remains unchanged. If a conflict occurs (i.e., $\chi(u)=\chi(v)$), then one needs to fix the coloring by recoloring one vertex from $\{u,v\}$, say $u$. Instead of scanning the whole neighborhood of $u$ to find the color (called a \emph{blank} color) that has not been used by any of its neighbors, the algorithm in~\cite{BCHN18:coloring} tries to \emph{sample} a color from the set that contains both blank colors and colors (called \emph{unique} colors) that have been used by exactly one neighbor of $u$. It is observed in \cite{BCHN18:coloring} that such a set has size $\Omega(\Delta)$, which guarantees that a future conflict edge incident to $u$ occurs with low probability (i.e., with probability at most $O(\frac{1}{\Delta})$). On the other hand, if a unique color is chosen, one needs to recolor the corresponding vertex $w$ (which is a neighbor of $u$), again, using a new color sampled from the set of blank and unique colors for $w$. This procedure might cause a cascade and even not terminate at all. %
The dynamic $(\Delta+1)$-vertex coloring algorithm of~\cite{BCHN18:coloring} resolves this problem by maintaining a hierarchical graph decomposition, and when recoloring a node it picks a color randomly out of all colors that are either (i) used by none of the neighbors or (ii) used by at most one of the neighbors {\em on a lower level} in the graph hierarchy.  The resulting algorithm is then shown to having $O(\log \Delta)$ amortized update time for maintaining a proper coloring. However, maintaining such a hierarchical partition is not only complicated, but also inefficient, as it alone already takes $O(\log \Delta)$ amortized update time. 
	
Now we briefly describe our main idea which leads to a constant-time dynamic coloring algorithm.	We show  that an approach based on \emph{assigning random ranks} to vertices outperforms the graph-hierarchy based algorithm: 
	During preprocessing each node $v$ is assigned a random rank  $r(v)$ from $[0,1]$ and a random color (assuming as usual that the initial graph is empty). Let $L_v$ denote the
	set of neighbors of a node $v$ with rank lower than $r(v)$ and let $\GoodL{v}$ denote the set of neighbors of $v$ whose rank is at most the median rank of the nodes in $L_v$.
	When recoloring  $v$, 
	we pick a color randomly out of all colors that are either (i)  used by none of its neighbors (called {\em blank} colors)  or (ii) by at most one neighbor in  $L_v$  {\em and} this node belongs to $\GoodL{v}$. (We show that there are always $\Omega(|L_v|)$ many such colors.)
	In case (ii) this neighbor $w$ must be recolored. 
	This is done with a more refined recoloring procedure 
	that additionally to the above information takes into account which nodes of $L_{w}$ also belong to $N(v)$, the neighborhood of $v$. It randomly samples  a color out of the set which consists of (i) all blank colors and (ii) all
	colors which are used by exactly one node
	in $L_w$ and, depending on the size of the respective sets,  {\em either} are used by a node in $L_{w,\textrm{new}}^{\textrm{<}}$, \emph{or} by a node in $L_{w, \textrm{old}}^{\textrm{<}}$, where $\NewL{w}:=L_{w} \setminus N(v)$ and $\OldL{w}:=L_{w} \cap N(v)$. 
	{More specifically, if $L_{w,\textrm{new}}^{\textrm{<}}$ combined with the blank colors
	has size $\Omega(|L_w|)$, we sample from this set, otherwise from $L_{w, \textrm{old}}^{\textrm{<}}$ combined with the blank colors.}
	This is necessary to guarantee that the new color is chosen randomly from a set of
	$\Omega(|L_w|)$ colors. 
	
	If the color of a node $y$
	in $L_{w,\textrm{new}}^{\textrm{<}}$ or $L_{w,\textrm{old}}^{\textrm{<}}$ 
	was chosen, $y$ will be recolored recursively taking $N(x)$ for {\em all} previously visited nodes $x$ into account. If $y$ was chosen from $L_{w,\textrm{new}}^{\textrm{<}}$, it is called a {\em good} vertex, otherwise a {\em bad} vertex.
	This results in a recoloring of nodes along a random path $P$ in the graph until a blank color is chosen (which is guaranteed to happen if a node $y$ with
	$L_y = \emptyset$ is reached) in total time $O(\sum_{y \in P} |L_y|)$.
	
	However, even though the rank of the next node is at most the median rank of the lower-ranked neighbors of  the previous node, 
	(which, if there were no dependencies between the ranks of the nodes on $P$, would imply that the {\em expected rank} will halve in each step),
	the {\em expected size of $L_y$} is not guaranteed to halve in each step.  
	To deal with this we need to (a) introduce a novel potential function
	$\Phi$ based on the sizes of a suitable subset of $L_y$  for each $y$ on $P$ and (b) {carefully analyze} the expected number of lower-ranked neighbors of the nodes on $P$, dealing with dependencies that arise by the fact that two nodes on $P$ might share neighbors.  
	More specifically, we show that,
	when traversing $P$ from an initial vertex $v$, at every good vertex the {\em expected rank} halves, while at every bad vertex $\Phi$ drops. 
	As (i) $\Phi$ is always non-negative, (ii) $\Phi$ only increases at good vertices, and (iii) the drop of $\Phi$ gives an upper bound of the time spent at bad vertices, we can {\em bound the total time for coloring all the vertices on $P$ by the total time spent at the good vertices on $P$}. At the good vertices, however, the expected rank is halved. 
	Due to our { sampling routine picking colors from neighbors with at most median rank} and with a careful analysis of the dependencies, we show that the total expected time at the good vertices on $P$, i.e. $O(\sum_{y \in P, y {\textrm{:good}}} |L_y|)$,
	forms a geometric series adding up to 
 $O(\alpha \Delta)$, where $\alpha$ is the rank of the initial vertex $v$. Finally, we combine this bound with the fact that for many operations (such as all deletions and many insertions) no recoloring is necessary to show that the expected amortized time per update operation is constant. 
 This 
 depends crucially on the fact 
 that the color of each node $v$ was picked uniformly at random from a set of $\Omega(|L_v|)$ many colors.
 
 Note that the refined sampling routine as well as the analysis that combines a potential function analysis with a careful analysis of the expected size of the sets $L_y$ along a random path $P$ is novel.
	Furthermore, while the idea of assigning random ranks to nodes was used before in the area of property testing (see e.g.~\cite{NO08:constant,HKNO:local,YYI12:constant,ORRR12:VC}) and in dynamic distributed algorithms~\cite{censor2016optimal}, this is, to the best of our knowledge,
	 its first use in  (centralized) dynamic graph algorithms for general graphs.
	The technique  has the advantage that, unlike in a hierarchical graph decomposition where the ordering of nodes by levels might change and needs to be updated, the ordering of nodes by ranks is static and does not create update costs. However, it has the disadvantage that, unlike in the hierarchical graph decomposition of~\cite{BCHN18:coloring},  (1) we do not have a worst-case upper bound on the number of nodes that are ``lower'' in the ordering
	and (2) the length of $P$, which is limited by the  longest strictly decreasing path in the ordering, might be $\Theta(n)$ and not $\Theta(\log \Delta)$ in the worst case, as 
	in~\cite{BCHN18:coloring}. 
We believe that this approach of assigning (static) random ranks to vertices instead of maintaining a hierarchical graph decomposition is of independent interest as
it might lead to constant-time
algorithms for other dynamic graph problems.

(2) {\em $(1+\varepsilon)$-approximation for the weight of an MSF of a graph $G$.} %
Both our deterministic and randomized algorithms for this problem use an approach developed in the area of property testing: Build an efficient algorithm for estimating the number of connected components (CCs) in a graph and apply it to suitable subgraphs of $G$~\cite{CRT05:MST,CS09:MST,AGM12:linear}.
More specifically, we build constant-time dynamic algorithms  that estimate the number of
CCs with appropriate additive error, apply them to $O(\log W/\varepsilon)$ many subgraphs,
and then use an extension of the formula in~\cite{CRT05:MST} to disconnected graphs to estimate the weight of an MSF. 
	
Though the above idea is quite simple, it is non-trivial to dynamically maintain the number CCs with the ``right'' additive error (see \emph{High-Level Ideas} in Section~\ref{sec:highlevel}). Furthermore, the techniques are very different from the fastest dynamic (exact) MSF algorithms: the algorithm of~\cite{HRW15:MST,HLT01:connectivity}
maintains a hierarchical decomposition with $O(\log n)$ levels, the algorithm of~\cite{NSW17:MST} maintains a decomposition of the graph into
expanders and a ``remaining'' part.

\subsection{Other Related Work}\label{sec:related_work}
Partially due to the $\Omega(\log n)$ lower bound for the fundamental problem of testing connectivity \cite{PatrascuD06}, a
large amount of previous research on dynamic graph algorithms has focused on algorithms with polylogarithmic or super-polylogarithmic update time.
Examples include testing $k$-edge (or vertex) connectivity~(see e.g., \cite{EppsteinGIN97,HLT01:connectivity,HenzingerK99}), maintaining minimum spanning tree~(see e.g., \cite{Fred83,EppsteinGIN97,HenzingerK99,HK97:mst,HLT01:connectivity,HRW15:MST,KapronKM13,wulff2017fully,NS17dynamic,NSW17:MST}), and graph coloring~\cite{barenboim2017fully,barba2017dynamic,BCHN18:coloring,SW18:coloring, duan2019dynamic}. There are also studies on \emph{partially dynamic} graph algorithms, including \emph{incremental algorithms} that only allow edge insertions, and \emph{decremental algorithms} that only allow edge deletions throughout all the updates. In contrast to such studies, our work is focusing on \emph{fully dynamic} algorithms, in which both edge insertions and deletions are allowed. 

We remark that the connections of dynamic graph algorithms with other fields that concern ``locality'' has been exploited. 
Such examples include Solomon's work on local algorithms for constructing bounded-degree sparsifiers, which lead to a dynamic algorithm for maintaining a $(1+\varepsilon)$-approximation of vertex cover with constant update time for any constant $\varepsilon>0$ in planar graphs~\cite{solomon2018local}. 
The classic sparsification technique for dynamic graph algorithms in some sense also makes uses of locality~\cite{EppsteinGIN97}. 
There is more work in dynamic graph algorithms that are explicitly or implicitly related to techniques from distributed computation~(e.g., \cite{barenboim2017fully}), not to mention the area distributed dynamic graph algorithms~(see e.g., \cite{censor2016optimal}). 
However, the locality properties used in these techniques are different from ours (i.e., property testing techniques), as the corresponding static property testing algorithm not only just read the local neighborhood of vertices, but just read such neighborhoods of \emph{very few} vertices, while still have provable performance guarantee of the global structure of graphs.  

In the field of constant-time algorithms, an algorithm is given \emph{query access} (e.g., to the adjacency list and/or adjacency matrix) to a fixed input graph (instead of a dynamic graph), and only makes a constant number queries to the graph. Examples include $(1+\varepsilon)$-approximating the weight of the MST \cite{CRT05:MST}, approximating the number of connected components with additive error $\varepsilon n$~\cite{CRT05:MST,BKM14:numcc}, $(1,\varepsilon n)$-approximating\footnote{We use the term $(\alpha,\beta)$-approximation algorithm to denote an algorithm that approximates the objective with a guarantee of $\alpha \text{Opt} + \beta$, where $\text{Opt}$ denotes the cost of an optimal solution.} the size of maximal/maximum matching in bounded average graphs~\cite{NO08:constant,YYI12:constant}, $(2,\varepsilon n)$-approximating the minimum vertex cover size~\cite{PR07:sublinear_distributed,MR09:distance,ORRR12:VC}, $(O(\log d), \varepsilon n)$-approximating the minimum dominating set size~\cite{PR07:sublinear_distributed,NO08:constant}. For $d$-bounded minor-free graphs, there are constant-time $(1,\varepsilon n)$-approximation algorithms for the size of minimum vertex cover, minimum dominating set and maximum independent set~\cite{HKNO:local}.   

Recently, constant-time property testing and approximation algorithms have also been transformed to constant-space random order streaming algorithms~\cite{MMPS17:testable,PS18:stream,CFPS19}. However, these work concentrates mainly on the space complexity of the algorithms and assumes the edges come from a \emph{uniform random} order (that only allows edge insertions), rather than an arbitrary edge sequence with both edge insertions and deletions as we are considering here. Sublinear-time algorithms for $(\Delta+1)$-vertex coloring in the graph streaming model, query access model and the massively parallel computation model have recently been studied~by Assadi et al.~\cite{ACK19:coloring}.

\section{Preliminaries}\label{sec:preliminaries}
A \emph{fully dynamic} graph algorithm is an algorithm that maintains a 
graph property in a graph $G=(V,E)$ which is undergoing an arbitrary sequence of the following  operations: 1) \Insert($u,v, w$): insert the edge $(u,v)$ with weight $w$ in $G$; 2) \Delete($u,v$): delete the edge $(u,v)$ from $G$.
If the considered graph is unweighted, then the weight $w$ of an insertion is always set to be $1$. 
In the $(\Delta+1)$-vertex coloring problem the fully dynamic graph algorithm maintains after each update operation a proper $(\Delta+1)$-vertex coloring of the current graph, where $\Delta$ is the maximum degree of any vertex since the beginning of the sequence until now. When asked to perform a \Query($u$) operation, the algorithm returns the color of the given vertex $u$. In the $(1+\varepsilon)$-approximate MSF weight problem the fully dynamic algorithm maintains a value $M'$ that is
a $(1+\varepsilon)$ approximation of the weight of the MSF in the current graph and it can return this value in constant time. 
It returns $M'$  when asked a \Query() operation.
\section{Maintaining a Proper $(\Delta+1)$-Vertex Coloring}\label{sec:coloring}
In this section, we give our constant-time dynamic algorithm and its analysis for maintaining a proper $(\Delta+1)$-coloring in a dynamic $\Delta$-bounded graph\footnote{Our algorithm can actually be extended to handle the case that the maximum degree $\Delta$ also changes. See Appendix~\ref{sec:chaningdelta} for more discussions.} and present the proof of Theorem~\ref{thm:coloring}. Recall that a dynamic graph is called to be $\Delta$-bounded if throughout the updates, it is $\Delta$-bounded. Given $\Delta$, let $\mathcal{C}:=\{1,\cdots, \Delta+1\}$ denote the set of {\em colors}. 
A coloring $\chi: V\rightarrow \mathcal{C}$ is \emph{proper} if 
$\chi(u)\neq \chi(v)$ for any $(u,v)\in E$. %

\subsection{Data Structures and the Algorithm}
\paragraph{Data structures.} We use the following data structures.

(1) We maintain a vertex coloring $\chi$ as an array such that $\chi(v)$ denotes the color of the current graph and guarantee that $\chi$ is a proper $(\Delta+1)$-vertex coloring after each update.

(2) For each vertex $v\in V$ we maintain: (a) its rank $r(v)$ that is chosen uniformly at random from $[0,1]$ during preprocessing; (b) its degree $\deg(v)$; (c) the last time stamp, denoted by $\tau_v$, at which $v$ was recolored; (d) two sets $L_v:=\{u: (u,v)\in E, r(u)<r(v)\},H_v:=\{u: (u,v)\in E, r(u)\geq r(v)\}$, which contain all neighbors of $v$ with ranks less than $v$, and all neighbors of $v$ with ranks at least $v$ (including $v$ itself), respectively; (e) the sizes of the previous two sets, i.e., $|L_v|$ and $|H_v|$. Note that $\deg(v)=|L_v\cup H_v|=|L_v|+|H_v|$.

{For each vertex $v\in V$ note that every color of $\cal C$ is either
(i) used by no neighbor of $v$ (and we call such color a \textit{blank} color for $v$), (ii) used by a neighbor in $H_v$, or (iii)
used by a neighbor in $L_v$ {and by {\em no} neighbor in $H_v$}. We call the corresponding sets of colors
(i) 
$\ColorB{v}$, (ii)
$\ColorH{v}$, and (iii) $\ColorLL{v}$.
We further partition $\ColorLL{v}$ into
(iii.1) $\ColorU{v}{L}$, which denotes the set of \textit{unique} colors for $v$ that have been used by exactly one vertex in $L_v$ and
(iii.2) $\ColorM{v}{L}$, which denotes the set of colors that have been used by at least two vertices in $L_v$.
Thus, $ \mathcal{C} = \ColorH{v}\  \dot\cup\  \ColorB{v}\ \dot\cup\ \ColorU{v}{L}\ \dot\cup\ \ColorM{v}{L}$. 
As it will be useful below, we finally define  $\ColorL{v}:=\ColorB{v}\cup \ColorU{v}{L}\cup \ColorM{v}{L}$.
Note that for any fixed $v$, a color $c$ can appear in exactly one of the two sets $\ColorH{v}$ and $\ColorL{v}$.}

(3) (i) For every vertex $v$, we maintain 
$\ColorH{v}$ and $\ColorL{v}$ in doubly linked lists. %
(ii) For each color $c\in \mathcal{C}$ and vertex $v\in V$, we keep the following information:
(a) a pointer $p_{c,v}$ from $c$ to its position in either $\ColorH{v}$ or $\ColorL{v}$, depending on which list it belongs to;
(b) a counter $\mu^H_v(c)$ such that $\mu^H_v(c)$ equals the number of neighbors in $H_v$ with color $c$ if $c\in \ColorH{v}$; or equals $0$ if $c\in \ColorL{v}$. %
(iii) For any vertex $v$ and color $c \in \mathcal{C}$ we keep the 
pointer $p_{c,v}$ in a hash table $\mathcal{A}_v$ which is indexed by $c$.
(iv) For any vertex $v$ and color $c\in \ColorH{v}$, we maintain the pairs $(c,\mu^H_v(c))$ in a hash table $\mathcal{A}^H_v$ which is indexed by the pair $(v,c)$.

More precisely, we use the dynamic perfect hashing algorithm by Dietzfelbinger et al.~\cite{DKMHRT94:hashing}, which takes amortized expected constant time per update and worst-case constant time for lookups. Alternatively we can get constant worst-case time for updates and lookups by spending time $O(n \Delta)$ during preprocessing to initialize suitable arrays. To simplify the presentation and since the randomness in the hash tables is independent of the randomness used by the algorithm otherwise, we will not mention the randomness introduced through the usage of hash tables in the following.
\paragraph{Initialization.} As the initial graph $G_0$ is empty, we initialize as follows:
(1) For each vertex $u\in V$, sample a random number (called \emph{rank}) $r(u)\in [0,1]$.
(2) Color each vertex $u$ by a random color $\chi(u)\in \mathcal{C}=\{1,\cdots, \Delta+1\}$ and initialize all the data structures suitably. \pnew{In particular, for each $u\in V$, we initialize $\ColorH{u}$ to be the empty list and $\ColorL{u}$ to be the doubly linked list containing all colors in $\mathcal{C}$. Note that the latter takes $O(n \Delta)$ time. In Appendix~\ref{subsec:initialization}, we show that the initialization time can be reduced to $O(n)$ while keeping constant expected amortized update time. }

{\bf Time stamp reduction.} Our algorithm does not use the actual values of the time stamps, only their relative order. Thus, every  $\poly(n)$ (say, $n^4$) number of updates we
determine the order of the vertices according to the time stamps and set the time stamps of every vertex to equal its position in the order and set the current time stamp to $n+1$. 
This guarantees that we only need to use $O(\log n)$ bits to store the time stamp $\tau_v$ for each vertex $v$ and it does not affect the ordering of the time stamps. 
The cost of the recomputation of the time stamps is $O(n \log n)$ and can be amortized over all the operations that are performed between two updates, increasing their running time only by an additive constant.

\paragraph{Handling an edge deletion.} Note that the edge deletion $(u,v)$ does not lead to a violation of the current proper coloring, so we do not need to recolor any vertex. We only need to update the data structures corresponding to the two endpoints, for which we refer to Appendix~\ref{sec:appendix_data_structure} for details.

\paragraph{Handling an edge insertion.} For an edge insertion $(u,v)$, we note that if $\chi(u)\neq \chi(v)$ before the insertion, then we only need to update the basic data structures corresponding to the two endpoints. If $\chi(u)=\chi(v)$, i.e, the current coloring $\chi$ is not proper any more, then we need to recolor one vertex $w\in\{u,v\}$ as well as to  update the relevant data structures. We always recolor the vertex that was colored last, i.e., the one with larger $\tau_w$. Wlog, we assume this vertex is $v$. Then we invoke a subroutine \textsc{Recolor}($v$) to recolor $v$ and potentially some other lower level vertices, and update the corresponding data structures. %
More precisely, we will first update $H_u, L_u, H_v, L_v$ and their sizes trivially in constant time. Then if $\chi(u)\neq \chi(v)$, we update the data structures corresponding to $u,v$ as described in Appendix~\ref{sec:appendix_data_structure}.

If $\chi(u)=\chi(v)$, and w.l.o.g., suppose that $\tau_v>\tau_u$, then we recolor $v$ by invoking the procedure \textsc{Recolor($v$)} below, where $\ColorU{v}{L}$ denotes the set of colors that have been used by {\em exactly one} vertex in $L_v$.
\begin{table}[h]
	\centering
	\begin{tabular}[t]{|p{\textwidth}|}
		\hline
		\textsc{Recolor($v$)}%
\begin{enumerate}
\item Run \textsc{SetColor}$(v)$ and obtain a new color $c$ (from $\mathcal{B}_v\cup \ColorU{v}{L}$).
\item Set $\chi(v)=c$. Update the data structures by the process $(\divideontimes)$ described in Appendix~\ref{sec:appendix_data_structure}.
\item\label{alg:recolor3} If $c\in \ColorU{v}{L}$,
\begin{enumerate}
\item Find the unique neighbor $w\in L_v$ with $\chi(w)=c$.
\item \textsc{Recolor($w$)}.
\end{enumerate}
\item {If $c\in \mathcal{B}_v$, then remove all the \textbf{visited} marks generated from the calls to \textsc{SetColor}.} 
\vspace{-1em}
\end{enumerate} 
		\\
		\hline
	\end{tabular}
\end{table}
Note that the recursive calls will eventually terminate as  for every recursive call \textsc{Recolor($w$)} in Step \ref{alg:recolor3} it holds that
$r(w) < r(v)$. Furthermore, no recursive call will be performed when $L_v = \emptyset$ as
it implies that $\ColorU{v}{L} = \emptyset$.
The subroutine \textsc{ReColor($v$)} calls the following subroutine \textsc{Setcolor($v$)}.

\begin{table}[h]
	\centering
	\begin{tabular}{|p{\textwidth}|}
		\hline
		\textsc{SetColor($v$)}
		\begin{enumerate}%
			\item Mark $v$ as \textbf{visited}. Initialize sets $\OldL{v}:=\{v\}$ and $\NewL{v}:=\emptyset$. 

			Scan the list $L_v$: for any $u\in L_v$, if it is marked as $\mathbf{visited}$, then add $u$ to $\OldL{v}$; otherwise (i.e., it is not marked), then add $u$ to $\NewL{v}$ and mark $u$ as $\mathbf{visited}$.

			\item\label{alg:recolor_low_degree_2} If $|L_v|+|H_v|< \frac{\Delta}{2}$ (i.e., $\deg(v)<\frac{\Delta}{2}$), repeatedly sample a color uniformly at random from $[\Delta+1]$ until we get a color $c$ that is contained in $\mathcal{B}_v$, the set of \emph{blank} colors for $v$ that have not been used by any neighbor of $v$.

			\item\label{alg:recolor_largeBv_2} Otherwise, we let $\GoodNewL{v}$ denote the subset of vertices in $\NewL{v}$ with ranks at most the median of all ranks of vertices in $\NewL{v}$. 
			{We let $\ColorU{v}{L_{\textrm{new}}^{\textrm{<}}}$ denote the set of colors that each has been used by exactly one vertex in $\NewL{v}$ and additionally this vertex belongs to $\GoodNewL{v}$.}  Define $\GoodOldL{v}$ and $\ColorU{v}{L_{\textrm{old}}^{\textrm{<}}}$ similarly.
 			\begin{enumerate}
				\item\label{alg:manynew} If $|\NewL{v}|\geq \frac{1}{10}|L_v|$ \pnew{or
				$L_v = \emptyset$}, then \pnew{we sample a random color $c$ from the set of the first $\min\{|\mathcal{B}_v\cup \ColorU{v}{L_{\textrm{new}}^{\textrm{g}}}|, |\GoodNewL{v}|+1\}$ elements of $\mathcal{B}_v\cup \ColorU{v}{L_{\textrm{new}}^{\textrm{<}}}$.}
				\item\label{alg:manyold} Else (i.e., $|\OldL{v}|>\frac{9}{10}|L_v|$) \pnew{we sample a random color $c$ from the set of the first $\min\{|\mathcal{B}_v\cup \ColorU{v}{L_{\textrm{old}}^{\textrm{<}}}|, |\GoodOldL{v}|+1\}$ elements of $\mathcal{B}_v\cup \ColorU{v}{L_{\textrm{old}}^{\textrm{<}}}$.}
			\end{enumerate}
			\item\label{alg:restore} Update the relevant data structures (i.e.~of $v$ and its neighbors in $L_v$) and \textbf{Return} $c$.   \vspace{-1em} 		
		\end{enumerate}\\
		\hline
	\end{tabular}
\end{table}

\subsection{The Analysis}
Next  we prove Theorem~\ref{thm:coloring}.
Let $v_0:=v$ be the vertex that needs to be recolored after an insertion and let $v_1, v_2, \cdots, v_\ell$ denote the vertices on which the recursive calls of \textsc{Recolor()}
were executed. We call $v_0,v_1,\cdots, v_\ell$ the \emph{recoloring path} originated from $v$. 
We first show that the expected total time for  all  calls \textsc{Recolor$(v_i)$} is $O(1+\sum_{i = 0}^\ell |L_{v_i}|)$. 
Then we
bound the expected value of this sum. 
We have the following lemma. %
\begin{lemma}\label{lemma:setcolor}
Subroutine \textsc{SetColor$(v)$} can be implemented to run in  $O(1+|L_v|)$ expected time, where the expectation is not over the random choices of ranks or colors at Step~\ref{alg:recolor_largeBv_2}, but comes from the use of hash tables and sampling colors at Step \ref{alg:recolor_low_degree_2}. For any recoloring path $v_0,v_1,\cdots,v_\ell$, 
the expected time for subroutine \textsc{Recolor($u$)} for any $u \in \{v_1, \dots , v_l\}$
excluding the recursive calls to \textsc{Recolor()} is 
$O(|L_u|)$ if $u\neq v_\ell$, and is
$O(1+\sum_{i = 0}^\ell |L_{v_i}|)$ if $u=v_\ell$. %
\end{lemma}

\begin{proof} %
	Recall that we store $L_v$, $\ColorH{v}$, and $\ColorL{v}$  for every vertex $v$.
	We use them to build all the sets needed in \textsc{SetColor$(v)$}.
	First we use an array $\RNewL{v}$ (resp. $\ROldL{v}$) to store ranks of vertices in $\NewL{v}$ (resp. $\OldL{v}$), and then find the median $\medianN{v}$ (resp. $\medianO{v}$) of the set of ranks of vertices in $\NewL{v}$ (resp. $\OldL{v}$) deterministically in $O(|\RNewL{v}|)=O(|L_v|)$ time~\cite{blum1973time}. 
	Traversing $L_v$ again (and using an empty array of length $\Delta$ that we clean again after this step) we compute (1) the sets $\ColorU{v}{L_{\textrm{new}}^{\textrm{<}}}$ and $\ColorU{v}{L_{\textrm{old}}^{\textrm{<}}}$ of colors that contain all colors that have been used by \emph{exactly one} vertex in $\GoodNewL{v}$, and by \emph{exactly one} vertex in $\GoodOldL{v}$, respectively,
	and (2) the sets $\ColorM{v}{L}$ of colors that contain all colors that have been used by \emph{at least two} vertices in $L_v$.
	Note that $\ColorU{v}{L} = \ColorU{v}{L_{\textrm{new}}^{\textrm{<}}} \cup \ColorU{v}{L_{\textrm{old}}^{\textrm{<}}}$, and, thus, it can be computed by copying these lists. 
	All these lists have size $O(|L_v|)$ and, thus, all these steps take time $O(|L_v|)$. 
	
	We will keep the sets $\ColorM{v}{L}$, $\ColorU{v}{L}$, $\ColorU{v}{L_{\textrm{new}}^{\textrm{<}}}$, $\ColorU{v}{L_{\textrm{old}}^{\textrm{<}}}$ in four separate lists and build hash tables for these sets with pointers to their positions in the lists. 
	Next we delete all colors in $\ColorM{v}{L}\cup \ColorU{v}{L}$ from the list $\ColorL{v}$ and the resulting list will be $\mathcal{B}_v$. 
	Note that the hash tables can be implemented in time linear in the size of corresponding sets, and each lookup (i.e., check if an element is in the set) takes constant worst-case time~\cite{DKMHRT94:hashing}.
	This completes the building of the data structure before Step 1.
	
	Recall that $|L_v| + |H_v| = \deg(v)$. Then for Step~\ref{alg:recolor_low_degree_2}, if $\deg(v)<\frac{\Delta}{2}$, we know that $|\mathcal{B}_v|>\Delta-\frac{\Delta}{2}= \frac{\Delta}{2}$. 
	Thus, a randomly sampled color from $[\Delta+1]$ belongs to $\mathcal{B}_v$ with probability at least $1/2$, which implies that in $O(1)$ expected time, we will sample a color $c$ from $\mathcal{B}_v$. 
	Note that a color $c$ belongs to $\mathcal{B}_v$ if and only if $c$ is not contained in $\ColorM{v}{L} \cup \ColorU{v}{L} \cup \ColorH{v}$, which can be checked by using the hash tables for $\ColorM{v}{L},$ for $ \ColorU{v}{L}$ and the hash table $\mathcal{A}^H_v$. 

	All the other steps only write, read and/or delete lists or hash tables of size proportional to $|L_v|$ or $|\ColorM{v}{L}\cup \ColorU{v}{L}|$, which is at most $|L_v|$. Though the list $\mathcal{B}_v\cup \ColorU{v}{L_{\textrm{new}}^{\textrm{<}}}$ might have size much larger than $|\GoodNewL{v}|$, it suffices to read at most $|\GoodNewL{v}|$ elements from it in Step \ref{alg:recolor_largeBv_2} (similar for $\mathcal{B}_v\cup \ColorU{v}{L_{\textrm{old}}^{\textrm{<}}}$ versus $|{\GoodOldL{v}}|$). In Step \ref{alg:restore}, to update the relevant data structures, we add all colors in $\ColorM{v}{L} \cup \ColorU{v}{L}$ back to the list $\mathcal{B}_v$ to construct $\ColorL{v}$.	
	
	To analyze the running time of \textsc{Recolor}($u$) (apart from the recursive calls), for any $u\in v_0,v_1,\dots,v_\ell$, note that apart from calling \textsc{Setcolor}($u$), \textsc{Recolor} updates the data structures, determines the neighbor $w$ that needs to be recolored next (if any) and if no such neighbor $w$ exists, i.e.~if $c$ is a blank color and $u$ is the last vertex of the recoloring path, then it unmarks all vertices that were marked by all the calls to \textsc{Setcolor} on the recoloring path. 
	For this \textsc{Setcolor} has stored all the marked vertices on a list, which it returns to \textsc{Recolor}. 
	This list is then used by \textsc{recolor} to unmark these vertices. 
	The time to update the data structures is constant expected time (the expectation arises due to the use of hash tables) to update its own
	data structure and $O(|L_u|)$ to update the data structures of its lower neighbors.
	Determining $w$ requires $O(|L_u|)$ time, as all lower neighbors of $u$ have to be checked. 
	Finally, \textsc{Recolor}($u$) for the last vertex $u=v_\ell$ on the recoloring path takes expected time $O(1+\sum_i |L_{v_i}|)$ as it unmarks all vertices on the recoloring path and their neighbors.
\end{proof}

Throughout the process we have two different types of randomness: one for sampling the ranks for the vertices and the other for sampling the colors. 
These two types of randomness are independent. 
Furthermore, only the very last vertex $v_\ell$ on the recoloring path $P = v_0,v_1,\cdots,v_\ell$ can satisfy the condition of Step~\ref{alg:recolor_low_degree_2} in \textsc{SetColor}, as once the condition is satisfied, we will sample a blank color which will not cause any further recursive calls. 
Thus, for all vertices on $P$, with the possible exception of $v_\ell$, Step~\ref{alg:recolor_largeBv_2} will be executed. 
We call a vertex $w$ with $\deg(w)<\frac{\Delta}{2}$ a \emph{low degree} vertex. 
Note that for a low degree vertex $w$, \textsc{SetColor}($w$) executes
Step~\ref{alg:recolor_low_degree_2} and takes $O(1)$ expected time, as with probability at least $1/2$ a randomly sampled color will be blank. 
In the following, we consider the expected time $T_v$ of recoloring
$P$ that excludes the time of recoloring any low degree vertex (which, if exists, must be the last vertex on $P$).  %
We first present a key property regarding the expected running time for recoloring a  vertex $v$. %
Let $N(v)$ denote the set of all neighbors of $v$ in the current graph. 
\begin{lemma}\label{lemma:recolor_time}
	Let $G$ denote the current graph. %
	For any vertex $v$ with rank $r(v)\leq\alpha$, the expected running time $T_v$ (over the randomness of choosing ranks of other vertices%
	) %
	is 
	\begin{eqnarray}
	\E[T_v | r(v)\leq \alpha]=
	O(\alpha \Delta) \label{eqn:expected_work_A} %
	\end{eqnarray}
Furthermore, conditioned on ranks of vertices in $N(v)$ and $r(v)\leq \alpha$, it holds that the expected running time $T_v$ (over the randomness of sampling ranks of $V\setminus (N(v)\cup\{v\})$) is %
	\begin{eqnarray}
	\E[T_v | r(v)\leq \alpha, r(w) \forall w\in N(v)]=
	O(|L_v|) + O(\alpha \Delta) \label{eqn:expected_work_B} %
	\end{eqnarray}	
\end{lemma}
The proof of the above lemma is deferred to Section \ref{sec:proof_recolor_time}. We will also need the following lemma regarding the size of the sampled color set. The proof of the lemma follows from a more refined analysis of the proof of Claim 3.1 in~\cite{BCHN18:coloring} and can be found in Appendix~\ref{subsec:app_coloring_proof}. %
\begin{lemma}\label{lemma:recolor}
Let $v$ be any vertex that needs to be recolored. Let $s$ denote the size of the set of colors that the algorithm samples from in order to choose a new color for $v$. Then it holds that 1) if $|L_v|+|H_v|<\frac{\Delta}{2}$, then $s\geq \frac{\Delta}{2}+1$; 2) otherwise, \pnew{$s\geq \frac{1}{100}|L_v|+1$.}
\end{lemma}

With the lemmas above, we are ready to prove Theorem~\ref{thm:coloring}.
\begin{proof}[\textbf{Proof of Theorem~\ref{thm:coloring}}]
	Note that an edge deletion does not lead to the recoloring of any vertex. Let us consider an insertion $(u,v)$. If $\chi(u)\neq \chi(v)$, we do not need to recolor any vertex. Otherwise, we need to recolor one vertex from $\{u,v\}$. Suppose w.l.o.g. that $\tau_v>\tau_u$, where $\tau_u$ denotes the last time that $u$ has been recolored. This implies that $v$ is recolored at the current time step, which be denote by $\tau$. We will invoke \textsc{Recolor($v$)} to recolor $v$. 
Recall that we let $T_v$ denote the running time of calling \textsc{Recolor}($v$), including all the recursive calls to \textsc{Recolor}, while \emph{excluding} the time of recoloring any low degree vertex (i.e.~a vertex where \textsc{SetColor}($w$) executed
Step~\ref{alg:recolor_low_degree_2}) and on the recoloring path originated from $v$ (which, if exists, must be the last vertex on the path). If the last vertex is indeed a low degree vertex, then the expected  total running time (over all sources of randomness) of \textsc{Recolor}($v$) will be $\E[T_v]+O(1)$, where the expectation $\E[T_v]$ in turn is over the randomness of sampling ranks of all vertices; otherwise, the expected total running time (over all sources of randomness) of \textsc{Recolor}($v$) will be $\E[T_v]$. %
	Let $\alpha_0=\frac{4C\log \Delta}{\Delta}$ for some constant $C\ge 1$. 
	Now we consider two cases:
	
	\textbf{Case I:} $r(v)\leq \alpha_0$. First we note that this case happens with probability at most $\alpha_0$ as $r(v)$ is chosen uniformly at random from $[0,1]$. Furthermore, by Lemma~\ref{lemma:recolor_time}, conditioned on the event that $r(v)\leq \alpha_0$, the expected time of the subroutine \textsc{Recolor}($v$) is $\E[T_v|r(v)\leq \alpha_0]= O(\alpha_0 \Delta)$, where the expectation is taken over the randomness of choosing ranks of all other vertices except $v$. %
	Therefore, the expected time of \textsc{Recolor}($v$) (over the randomness of choosing ranks of all vertices) %
	is at most $\alpha_0 \cdot O(\alpha_0 \Delta) = O(\alpha_0^2 \Delta)=O(\frac{\log^2 \Delta}{\Delta})=O(1)$.

	\textbf{Case II:} $r(v)>\alpha_0$. Let $r(v)=\alpha$. Conditioned on the event that $r(v)=\alpha$, by Lemma~\ref{lemma:recolor_time}, the expected running time (over the randomness of choosing ranks of other vertices) of \textsc{Recolor}($v$) at time $\tau$ %
	is $O(\alpha \Delta)$.%
	
	We let $L_v$ and $L'_v$ denote the set of neighbors of $v$ with ranks lower than $v$ in the graph at (current) time $\tau$ and at time $\tau_v$, (the latest time that $v$ was recolored), respectively. We define $H_v,H'_v$ similarly. %
	We let $\deg(v)=|L_v\cup H_v|$ and $\deg'(v)=|L'_v\cup H'_v|$ denote the degree of $v$ at time $\tau$ and $\tau_v$, respectively.
	
	Case (a): $\deg'(v)<\Delta/2$. 
	In this case, we know that at time $\tau_v$, we will sample a color from the set of blank colors $\mathcal{B}(v)$, which has size at least $\Delta/2$. 
	Thus, the probability that we sampled the color $\chi(u)$ is at most $2/\Delta$. On the other hand, at time $\tau$, we will spend at most $O(\alpha\Delta)=O(\Delta)$ expected time (over the randomness of sampling ranks of vertices in $V\setminus\{v\}$). %
	Thus, the expected time (over the randomness of sampling ranks and of sampling colors at time $\tau_v$) we spent on recoloring $v$ at time $\tau$ is $O(\frac{1}{\Delta} \cdot \Delta)=O(1)$.%
	
	Case (b): $\deg'(v)\geq \Delta/2$. We now consider two sub-cases. 		
	
	Case (b1): If $\deg(v)<\Delta/4$, %
	then there must have been at least $\deg'(v)/2=\Omega(\Delta)$ deletions
	of edges incident to $v$ between $\tau_v$ and $\tau$.
	We can recolor $v$ at time $\tau$ in expected $O(\alpha\Delta)=O(\Delta)$ time. We charge this time to the updates incident to $v$ between $\tau_v$ and $\tau$. Note that each update is only charged twice in this way, once from each endpoint, adding a constant amount of work to each deletion. 
	
	Case (b2): If $\deg(v)\geq\Delta/4$, %
	then $\E[|L_v|]=\alpha \deg(v)\geq \alpha\Delta/4\geq \frac{\alpha_0\Delta}{4} \geq C\log \Delta$ for some constant $C \ge 1$ and $\E[|L_v|]=\alpha \deg(v)\leq \alpha\Delta$. Then over the randomness of sampling ranks for vertices in $N(v)$, it follows from a Chernoff bound that with probability at least $1-\frac{1}{\Delta}$, $\frac{\E[|L_v|]}{2} \leq |L_v|\leq \frac{3\E[|L_v|]}{2}$, which implies that with probability at least
	$1 - 1/\Delta$ it holds that
	$$\frac{\alpha \Delta}{8}\leq \frac{\E[|L_v|]}{2} \leq |L_v|\leq \frac{3\E[|L_v|]}{2}\leq \frac{3\alpha \Delta}{2}$$
	By inequality (\ref{eqn:expected_work_B}) in Lemma~\ref{lemma:recolor_time}, over the randomness of sampling ranks for $V\setminus (N(v)\cup \{v\})$, the expected work for recoloring $v$ at time $\tau$ is $O(|L_v|)+O(\alpha\Delta) = O(\alpha \Delta)$.
	Then the work for recoloring is
	$O(\Delta)$ as $|L_v| \le \Delta$. 
	Thus the expected work of this case is $\frac{1}{\Delta} \cdot O(\Delta)  = O(1)$.
	
	Next we analyze the case that the above inequalities hold and further distinguish two sub-cases.
	
	Case (b2-1): If $|L_v \triangle L'_{v}| > \frac{1}{10}|L_v|$, then there must have been at least $\frac{1}{10}|L_v|=\Theta(\alpha \Delta)$ edge updates incident to $v$ between $\tau_v$ and $\tau$. By the same argument as above we can amortize the expected work of $O(\alpha \Delta)$ over these edge updates, charging each edge update at most twice.
	This adds an  expected amortized cost of $O(1)$ to each update. 
	
	Case (b2-2): If $|L_v \triangle L'_{v}| \leq \frac{1}{10}|L_v|$, then it holds that $|L'_{v}| \ge |L_v| - |L_v \triangle L'_{v}| \geq \frac{9}{10} |L_v|$. 
	By Lemma~\ref{lemma:recolor}, $\chi(v)$ was picked at time $\tau_v$ from a set of $\Omega(|L'_{v}|)$ many colors. 
	Thus, the probability that we picked the color $\chi(u)$ at time $\tau_v$ is at most $O(\frac{1}{|L'_{v}|})=O(\frac{1}{|L_v|})$. 
	As the expected work at time $\tau$ is at most $O(\alpha \Delta) =O(|L_v|)$ (with the expectation over randomness of sampling ranks), the expected amortized update time is $O(\frac{1}{|L_v|})\cdot O(|L_v|)=O(1)$. 
	
	This completes the proof of the theorem.
\end{proof}	
\subsubsection{Bounding the Expected Work per Recoloring: Proof of Lemma~\ref{lemma:recolor_time}.}\label{sec:proof_recolor_time}
Let %
 $v_0, v_1, \cdots$ be the vertices on the recoloring path after an insertion.
 By Lemma~\ref{lemma:setcolor} the total expected time for all calls \textsc{Recolor($v_i$)}
is $O(1+\sum_{i \ge 0} |L_{v_i}|)$. Recall that the running time $T_v$ excludes the time spent on recoloring a low degree vertex (and a low degree vertex can only be the last vertex of a recoloring path). Thus, for all vertices $v_i$ that contribute to $T_v$ only Step~\ref{alg:manynew} or Step~\ref{alg:manyold} of \textsc{SetColor} can occur. 
 Let $v_{i_0}=v_0, v_{i_1}, v_{v_2}, \cdots$ be the vertices for which Step \ref{alg:manynew} occurred during \textsc{Setcolor($v$)}, which we call \emph{good} vertices. We bound the expected value of ranks of good vertices and the expected size of the lower-ranked neighborhood of these vertices in the following lemma. Note that the expectations are taken over the randomness for sampling ranks of vertices, whose ranks are \emph{not} in the conditioned events.%
\begin{lemma}\label{lemma:ranks}
	For any $j\geq 0$, it holds that 
	$$\E[r(v_{i_j+1})|r(v_0)\leq \alpha] \leq \frac{\alpha}{2^j}, \quad \E[|L_{v_{i_j}}|\ | r(v_0)\leq \alpha]\leq  \frac{10\cdot \alpha\cdot \Delta}{2^{j-1}}.$$
	Furthermore, for any $j\geq 1$, it holds that
	$$\E[r(v_{i_j+1})|r(v_0)\leq \alpha, r(w) \forall w\in N(v_0)] \leq \frac{\alpha}{2^{j-1}}, \quad \E[|L_{v_{i_j}}|\ | r(v_0)\leq \alpha, r(w) \forall w\in N(v_0)]\leq  \frac{10\cdot \alpha\cdot \Delta}{2^{j-2}}.$$
\end{lemma}
\begin{proof}
	To prove the lemma, we use the principle of deferred decisions: Instead of sampling the ranks for all vertices (independently and uniformly at random from $[0,1]$) at the very beginning, we sample the ranks of vertices sequentially by the following random process: 

Starting from $v_0$ with rank $r(v_0)$, we sample all the ranks of vertices in $N(v_0)$. We will then choose $v_1$ as described in the algorithm \textsc{Recolor} (if a non blank color has been sampled). Now for each $i\geq 1$, we note that the ranks of all the vertices in 
		$\OldN{v_i}:=N(v_i)\cap (\cup_{j<i} N(v_j) \cup \{v_0\})$
		have already been sampled, and then we only need to sample (independently and uniformly at random from $[0,1]$) the ranks for all vertices in 
		$\NewN{v_i}:=N(v_i)\setminus \OldN{v_i}.$
		In this case, we say that the ranks of vertices in $\NewN{v_i}$ are sampled \emph{when we are exploring $v_i$}. Then we will choose $v_{i+1}$ in the algorithm \textsc{Recolor} (if a non blank color has been sampled). We iterate the above process until \textsc{Recolor} has sampled a blank color.

	For any $i$, we call $\NewN{v_i}$ the \emph{free neighbors} of $v_i$ with respect to $v_0,v_1,\cdots, v_{i-1}$. In particular, $\NewN{v_0}=N(v_0)$ and $N(v_i) = \NewN{v_i} \dot\cup \OldN{v_i}$.
	Now a key observation is that
	\begin{itemize}
		\item[$(\star)$] for any vertex $v_i$, it holds that $\NewL{v_i}$ (as defined in the algorithm \textsc{SetColor}($v_i$)) is entirely determined by {the ranks of} the vertices $\NewN{v_i}$ and is independent of the randomness for sampling ranks of $\OldN{v_i}$.
	\end{itemize} 
	This is true since $\NewL{v_i}$ contains all the neighbors of $v_i$ with ranks less than $r(v_i)$ and have not been visited so far: for any vertex in $\OldN{v_i}$, either its rank is higher than $v_i$, or its rank is less than $v_i$ and it has been marked as \textbf{visited} before we invoke \textsc{SetColor}($v_i$).

{We first prove the first part of the lemma.} We {\em assume for now that $r(v_0)$ is fixed} and we denote by $\Ra{i_j}$ the randomness of sampling ranks for vertices in $\NewN{v_{i_j}}$. %
We will prove by induction on the index $j$ that 
\begin{eqnarray}
\E_{\Ra{i_j}}[r(v_{i_j+1})] \le \frac{r(v_0)}{2^j} \text{    and        } \E_{\Ra{i_j}}[|\NewL{v_{i_{j}}}|] \leq \frac{r(v_0)\cdot\Delta}{2^{j-1}}.\label{ineq:induction1}
\end{eqnarray}
Note that this holds for $j=0$ since $i_0 = 0$,  $r(v_1) \le r(v_0)$, 
$\NewL{v_{i_0}} = L_{v_0}$, and $\E_{\Ra{0}}[|L_{v_0}|] =r(v_0) \cdot |N(v_0)| \leq  r(v_0) \cdot \Delta$.
Next we assume it holds for $j-1$, and prove it also holds for $j$. 
	By the definition of the good vertex $v_{i_j}$, we know that $v_{i_j+1}\in L_{v_{i_j}}$, and that the rank of $v_{i_j+1}$ is at most the median, denoted by $\medianN{v_{i_j}}$, of all the ranks of vertices in $\NewL{v_{i_j}}$, which in turn consists of all vertices in $\NewN{v_{i_j}}$ with rank not larger than $r(v_{i_j})$. %
	Furthermore, by the observation $(\star)$, the rank of $r(v_{i_j+1})$ depends only on $r(v_{i_{j}})$ and the ranks in $\NewN{v_{i_j}}$. This implies that 
	$$\E_{\Ra{i_j}}[r(v_{i_j+1})|r(v_{i_j})]\leq \E_{\Ra{i_j}}[\medianN{v_{i_j}}|r(v_{i_j})]\leq \frac{r(v_{i_j})}{2},$$
	where the last inequality follows from the fact that $\medianN{v_{i_j}}$ is the median of a set of numbers chosen independently and uniformly at random from $[0,1]$, conditioned on that they are at most $r(v_{i_j})$ (see e.g., Lemma 8.2 and 8.3 in~\cite{MU05:probability}).  %
	Since $r(v_{i_{j}})\leq r(v_{(i_{j-1})+1})$ in all cases and, by the induction assumption, $\E_{\Ra{i_{j-1}}}[r(v_{(i_{j-1})+1})] \leq \frac{r(v_0)}{2^{j-1}}$, it holds that %
\begin{eqnarray*}
\E_{\Ra{i_j}}[r(v_{i_j+1})]
\leq\E_{r(v_{i_j})}[\E_{\Ra{i_j}}[r(v_{i_j+1})|r(v_{i_j})]]
&\leq&\frac{1}{2} \E_{r(v_{i_j})}[r(v_{i_{j}})]\\
 &\leq& \frac12 \E_{\Ra{i_{j-1}}}[\E_{r(v_{i_j})}[r(v_{i_j})|r(v_{(i_{j-1})+1})]]\\
 &\leq& \frac{1}{2}\E_{\Ra{i_{j-1}}}[r(v_{(i_{j-1})+1})] 
\leq 
\frac{r(v_0)}{2^{j}}.
\end{eqnarray*}

	Furthermore, for any $j \ge 0$, by the observation $(\star)$, $\NewL{v_{i_j}}$ depends only on $r(v_{i_{j}})$ and ranks in $\NewN{v_{i_j}}$. Thus
	$$\E_{\Ra{i_j}}[|\NewL{v_{i_{j}}}|\ |r(v_{i_j})]\leq r(v_{i_j}) \cdot|\NewN{v_{i_j}}|\leq  r(v_{i_j})\cdot\Delta.$$
	This further implies that %
	\begin{eqnarray*}
		\E_{\Ra{i_j}}[|\NewL{v_{i_{j}}}|] =
		\E_{r(v_{i_j})}[\E_{\Ra{i_j}}[|\NewL{v_{i_{j}}}|\ | r(v_{i_j})]]
		\leq  \E_{r(v_{i_j})}[r(v_{i_j})]\cdot\Delta %
		\leq \frac{r(v_0)\cdot\Delta}{2^{j-1}}.
	\end{eqnarray*}
	
	Now let us no longer assume that $r(v_0)$ is fixed, but instead condition on the event that $r(v_0) \le \alpha$.
	Then it follows that 
	$\E_{\Ra{i_j}}[r(v_{i_j+1})| r(v_0)\leq \alpha] \leq \frac{\alpha}{2^{j}} \text{ and }
		\E_{\Ra{i_j}}[|\NewL{v_{i_{j}}}|\ |r(v_0)\leq \alpha] \leq \frac{\alpha\cdot\Delta}{2^{j-1}}.$
	
	Finally, by the definition of good vertices, it holds that $|\NewL{v_{i_{j}}}|\geq \frac{1}{10} |L_{v_{i_{j}}}|$. This implies that 
	\begin{eqnarray*}
		\E_{\Ra{i_j}}[|L_{v_{i_{j}}}|\ |r(v_0)\leq \alpha] &\leq& 10\cdot\E_{\Ra{i_j}}[|\NewL{v_{i_{j}}}|\ |r(v_0)\leq \alpha] \\
		&\leq& 10\cdot\frac{\alpha\cdot\Delta}{2^{j-1}}.
	\end{eqnarray*}
	This completes the proof of the first part of the lemma.
	
{For the ``Furthermore'' part of the lemma, the analysis is similar as above. Now we start with the assumption that $r(v_0),r(w)\forall w\in N(v_0)$ are fixed. Note that $v_{i_1}\in N(v_0)$, which implies that $r(v_{i_1})$ is also fixed. We will then prove by induction on the index $j$ that $$\E_{\Ra{i_j}}[r(v_{i_j+1})] \le \frac{r(v_{i_1})}{2^{j-1}} \text{ and } \E_{\Ra{i_j}}[|\NewL{v_{i_{j}}}|] \leq \frac{r(v_{i_1})\cdot\Delta}{2^{j-2}}.$$
	In the base case $j=1$, the above two inequalities hold as $r(v_{i_1+1})\leq r(v_{i_1})$ and $\E_{\Ra{i_1}}[|\NewL{v_{i_1}} |] = r(v_{i_1}) \cdot  |\NewN{v_{i_1}}|\leq r(v_{i_1})\cdot \Delta$. {The inductive step from case $j-1$ to $j$ can be then proven in the same way as we proved Inequalities (\ref{ineq:induction1})}. 
	Then instead of assuming that $r(v_0),r(w)\forall w\in N(v_0)$, we condition on the event that $r(v_0)\leq \alpha,r(w)\forall w\in N(v_0)$, which directly implies that $r(v_{i_1})\leq \alpha$. Then it follows that $\E_{\Ra{i_j}}[r(v_{i_j+1})|r(v_0)\leq \alpha,r(w)\forall w\in N(v_0)] \le \frac{\alpha}{2^{j-1}} \text{ and } \E_{\Ra{i_j}}[|\NewL{v_{i_{j}}}|\ |r(v_0)\leq \alpha,r(w)\forall w\in N(v_0)] \leq \frac{\alpha\cdot\Delta}{2^{j-2}}.$ 
Finally, by the definition of good vertices, $|\NewL{v_{i_{j}}}|\geq \frac{1}{10} |L_{v_{i_{j}}}|$, which implies that $\E_{\Ra{i_j}}[|L_{v_{i_{j}}}|\ |r(v_0)\leq \alpha,r(w)\forall w\in N(v_0)]\leq 10  \cdot \E_{\Ra{i_j}}[|\NewL{v_{i_{j}}}|\ |r(v_0)\leq \alpha,r(w)\forall w\in N(v_0)] \leq \frac{10\alpha\cdot\Delta}{2^{j-2}}.$ This completes the ``Furthermore'' part of the lemma.}
\end{proof}

Now we relate the total work to the work incurred by Step~\ref{alg:manynew}. Note that the total work $T_v$ is proportional to the sum of sizes of all lower-ranked neighborhoods of $v_0,v_1,\dots$. We will prove the following lemma, which implies that the total work of recoloring $v$ is at most a constant factor of the total work for recoloring all the {\em good vertices} on the recoloring path. %
\begin{lemma}\label{lem:lowervsgoodlower}
	It holds that $\sum_i |L_{v_i}|\leq 3 \sum_{\textrm{$i: v_i$ is good}} |L_{v_i}| = 3\sum_{j} |L_{v_{i_j}}|$.
\end{lemma}
\begin{proof}
	We first introduce the following definition. For any $i$ and $k<i$, we let $\FirstLower(v_k,v_i)$ denote the set of vertices whose ranks are less than $r(v_i)$, and are sampled when we are exploring $v_k$, i.e., $\FirstLower(v_k,v_i)=\{w: w\in \NewN{v_k}, r(w)<r(v_i) \}$. Note that as $r(v_{i+1})<r(v_{i})$, it always holds that for any $0\leq k< i$, 
	$\FirstLower(v_k,v_{i+1})\subseteq \FirstLower(v_k,v_{i}).$	
	Now we define the following potential function $\Phi$:
	\begin{eqnarray}
	\Phi(-1) := 0 \text{ and }\ \Phi(i):=\sum_{k:k\leq i} |\FirstLower(v_k,v_{i+1})| \ \ \forall i \ge 0, %
	\label{eqn:monotone}
	\end{eqnarray}
	We have the following claim regarding the potential functions. %
	\begin{claim}\label{claim:1}
For any $i \le 0$, $\Phi(i)\geq 0$. Furthermore, 
if $v_i$ is a good vertex, then $\Phi(i) - \Phi(i-1) \le |L_{v_i}|/2$,
otherwise $\Phi(i) - \Phi(i-1) \le - 7 |L_{v_i}|/20.$
	\end{claim}
\begin{proof}%
	Note that if Step \ref{alg:manynew} in subroutine \textsc{SetColor} is executed at vertex $v_i$, i.e., $v_i$ is good, then the potential $\Phi(i)$ might be larger or smaller than $\Phi(i-1)$. 
	If $v_i$ is good then $|\FirstLower(v_i,v_{i+1})| \le \frac{|\GoodNewL{v_i}|}{2} $ by the fact that $r(v_{i+1})$ is at most the median rank in $\GoodNewL{v_i}$.
	Furthermore, it holds that 
	\begin{eqnarray*}
	\Phi(i) = \sum_{k:k\leq i} |\FirstLower(v_k,v_{i+1})|
	&\leq& \sum_{k:k\leq i-1} |\FirstLower(v_k,v_{i})| + |\FirstLower(v_{i},v_{i+1})| \nonumber\\
	&\le&\Phi({i-1})+\frac{|\GoodNewL{v_i}|}{2} 
	\leq\Phi({i-1}) + \frac{|L_{v_i}|}{2}~\label{eqn:phi_manynew}
	\end{eqnarray*}
	
	Now suppose that Step \ref{alg:manyold} is executed at vertex $v_i$, i.e., $v_i$ is not good. %
	Since $v_{i+1}$ is a vertex from the lower half of the old lower neighbors of $v_i$ (i.e., $v_{i+1}\in \GoodOldL{v_i} \subseteq \cup_{k<i} \FirstLower(v_k,v_{i}) \cap \OldL{v_i}$), we have that to obtain the set $\cup_{k<i} \FirstLower(v_k,v_{i+1})$ from the set $\cup_{k<i} \FirstLower(v_k,v_{i})$, we need to remove at least $\frac{1}{2}|\OldL{v_i}|\geq \frac{1}{2} (1-\frac{1}{10}) |L_{v_i}|$ vertices. Furthermore, $\FirstLower(v_i,v_{i+1})$ can contain at most $|\NewL{v_i}|\leq \frac{1}{10}|L_{v_i}|$ vertices. This implies that  
	\begin{eqnarray*}
	\Phi(i) = \sum_{k:k\leq i} |\FirstLower(v_k,v_{i+1})| 
	&=& \sum_{k:k\leq i-1} |\FirstLower(v_k,v_{i+1})| + |\FirstLower(v_i,v_{i+1})|\nonumber\\
	&\leq& \sum_{k:k\leq i-1} |\FirstLower(v_k,v_{i})|- \frac{1}{2} (1-\frac{1}{10}) |L_{v_i}| + \frac{1}{10}|L_{v_i}| \nonumber \\
	&=&\Phi({i-1})-\frac{7}{20}\cdot |L_{v_i}| \label{eqn:phi_manyold} 
	\end{eqnarray*}
\end{proof}

	Now we distinguish three types of indices. We call an index $i$, 
	 a \emph{type I} index, if  Step \ref{alg:manynew} occurred during
	\textsc{Setcolor($v$)} {\em and} $\Phi(i) - \Phi(i-1) \ge 0.$ 
	By Claim \ref{claim:1} it holds that for such an index $i$,
			$|L_{v_i}|\geq 2 (\Phi(i)-\Phi({i-1})).$
	 We call $i$ a \emph{type II} index, if 
	Step \ref{alg:manynew} occurred during
	\textsc{Setcolor($v$)} {\em and} $\Phi(i) - \Phi(i-1) \le 0.$
	 It holds that for such an index $i$ (as for any index), $|L_{v_i}|\geq 0$.
	 We call $i$ a \emph{type III} index, if Step \ref{alg:manyold} occurred during
	\textsc{Setcolor($v$)}, i.e. $v_i$ is not a good vertex. By Claim \ref{claim:1} it holds that for such an index $i$, $\Phi$ decreases and
		\begin{eqnarray*}
			|L_{v_i}| \leq (\Phi({i-1})-\Phi({i})) \cdot \frac{20}{7} < 3\cdot(\Phi({i-1})-\Phi({i})).
		\end{eqnarray*}

Now we bound the sum of sizes of lower-ranked neighborhoods of vertices corresponding to Step \ref{alg:manyold}. It holds that 	
\begin{eqnarray*}
		\sum_{\textrm{$i$: Step \ref{alg:manyold}}} |L_{v_i}| 
		\leq \sum_{\textrm{$i$: type III}} 3(\Phi({i-1})-\Phi(i)) 
		&\leq& \sum_{\textrm{$i$: type II or III}} 3(\Phi({i-1})-\Phi(i)) \\
		 \leq  \sum_{\textrm{$i$: type I}}  3(\Phi(i)- \Phi({i-1})) 
		& \leq &  \sum_{\textrm{$i$: type I}} 3 \cdot \frac{1}{2}|L_{v_i}|
		< \sum_{\textrm{$i$: type I}} 2 |L_{v_i}|
	\end{eqnarray*}
	where the third inequality follows from the fact that  $\Phi$ starts at 0 and is non-negative at the end, and, thus, the total decrease of $\Phi$ is at most its total increase.
	Thus, it follows that 
	\begin{eqnarray*}
		\sum_{i} |L_{v_i}| = 
		\sum_{\textrm{$i$: type I}} |L_{v_i}| +
		\sum_{\textrm{$i$: type II}} |L_{v_i}| +
		\sum_{\textrm{$i$: type III}} |L_{v_i}| 
		\leq 3\sum_{\textrm{$i$: type I or II}} |L_{v_i}| =3\sum_{j} |L_{v_{i_j}}| 
	\end{eqnarray*}
\end{proof}

Now we finish the proof of Lemma~\ref{lemma:recolor_time}. By Lemma \ref{lem:lowervsgoodlower} and Lemma~\ref{lemma:ranks}, it holds that
\begin{eqnarray*}
\E[\sum_i |L_{v_i}|\ | r(v)\leq \alpha] \leq 3\cdot\E[\sum_{j} |L_{v_{i_j}}|\ | r(v)\leq \alpha]
	= O(\alpha\cdot \Delta \cdot \sum_j \frac{1}{2^j}) =
	O(\alpha \Delta).
\end{eqnarray*}
Since the expected work $T_v$ satisfies that $T_v= O(\sum_i |L_{v_i}|)$, the first part of the lemma follows.
By the ``Furthermore'' part of Lemma~\ref{lemma:ranks}, it holds that
\begin{eqnarray*}
\E[\sum_i |L_{v_i}|\ | r(v)\leq \alpha, r(w) \forall w\in N(v)] 
	&\leq& 3\cdot |L_v| + 3\cdot\E[\sum_{j\geq 1} |L_{v_{i_j}}|\ | r(v)\leq \alpha, r(w) \forall w\in N(v)]\\
\leq 3\cdot |L_v|+3\cdot 10\cdot \alpha\cdot \Delta \cdot \sum_{j}\frac{1}{2^{j-2}}
	&=& 3\cdot |L_v| +O(\alpha\cdot \Delta \cdot \sum_j \frac{1}{2^j}) =
	O(|L_v|)+O(\alpha \Delta).
\end{eqnarray*}
Then the ``Furthermore'' part of Lemma \ref{lemma:recolor_time} follows from the fact that  $T_v= O(\sum_i |L_{v_i}|)$.

\setlist[enumerate]{itemsep=3pt,topsep=3pt}
\section{Maintaining the Approximate Weight of the MSF}\label{sec:parameters} 
In this section, we present our dynamic algorithms for maintaining the weight $M$ of a minimum spanning forest of 
a graph $G$ without
 parallel edges and with edge weights in $[1,W]$. Our algorithms exploit a relation between the weight of MSF of a graph $G$ and the number of CCs of some subgraphs of $G$. Let $G^{(\ell)}$ denote the subgraph of $G$ spanned by all edges with weights at most $\ell$ and let $c^{(\ell)}$ denote the number of CCs in $G^{(\ell)}$. 
 We will make use of the following Lemma. 
 \begin{lemma}[\cite{CRT05:MST,CS09:MST,AGM12:linear}]\label{lem:mst_ncc}
 	Let $G$ be a weighted graph\footnote{We remark that in \cite{CRT05:MST,CS09:MST,AGM12:linear}, the input graph is assumed to be \emph{connected}. The lemma we stated here does not require the connectedness assumption and its proof is a straightforward generalization of previous proofs, by noting that $M=n-W\cdot c^{(W)}+\sum_{i=0}^{W-1}c^{(i)}$ in a general graph $G$ with edges weights from $\{1,\dots, W\}$ for any integer $W\geq 1$.} and let $M,W,c^{(\ell)}$ be defined as above. Let $r=\lceil\log_{1+\frac{\varepsilon}{2}} W\rceil$. Let $\lambda_i=(1+\frac{\varepsilon}{2})^{i+1}-(1+\frac{\varepsilon}{2})^i$ and $\ell_i=(1+\frac{\varepsilon}{2})^i$. Then it holds that
 	\begin{eqnarray}
 	M\leq X:=n-c^{(\ell_r)} \cdot (1+\frac{\varepsilon}{2})^r + \sum_{i=0}^{r-1} \lambda_i\cdot c^{(\ell_i)}\leq (1+\frac{\varepsilon}{2}) M. \label{eqn:app_mst}
 	\end{eqnarray}	
 \end{lemma}

Now we first present some high-level ideas of the algorithms. 
\subsection{High-Level Ideas}\label{sec:highlevel}
As briefly mentioned in the introduction, we would like to build constant-time dynamic algorithms that estimate the number of CCs with appropriate additive error, apply them to $O(\log W/\varepsilon)$ many subgraphs,
and then use the formula (\ref{eqn:app_mst}) to disconnected graphs to estimate the weight of an MSF. In particular, we want to estimate the number of CCs with an additive error $\varepsilon'\cdot \nis(G)$, where $\nis(G)$ is the number of \emph{non-isolated} vertices in $G$ (see below why this is crucial). Our randomized dynamic algorithm for this problem achieves such an error in time 
$O(\max\{1,\frac{\log(1/\varepsilon')\log n}{(\varepsilon')^3\cdot m^*}\})$ with high probability (see precise statement in Section~\ref{subsec:randomized}), and our deterministic algorithm achieves the same error $\varepsilon' \cdot \nis(G)$ in time $O((1/\varepsilon')^2)$.

The \emph{randomized} algorithm 
uses the following general approach used before (see e.g.~\cite{GuptaP13}):
Whenever (1) there exists a static algorithm that in time $T$ estimates a desired parameter (here the number of CCs) with an additive error of $\err$
and (2) 
each update operation changes the value of a desired parameter only by an additive value up to $+/- \delta$ (here 1),
then  running the static algorithm every $\frac{\err}{\delta}$ update
operations leads to a dynamic algorithm with additive error of at most $2 \err$ and {\em amortized} time $O(\frac{T \delta}{\err})$ per update
and this can be turned into a {\em worst-case} time bound using ``rebuilds in the background''.
We use the static (constant-time) algorithm of~\cite{BKM14:numcc} (that improves upon~\cite{CRT05:MST}) for estimating the number of CCs 
with additive error $\varepsilon' n$ as a subroutine. By a straightforward application of the above general approach, we can obtain a dynamic estimator for the number of CCs with an additive error $\varepsilon' n^{2/3}\log^{2/3}n$ with $O(1/\varepsilon'^{3})$ update time. 

However, to use this algorithm for dynamically estimating the number of CCs with an additive error $\varepsilon'\cdot \nis(G)$ achieving the above bound, we need to carefully choose different values of $\err$ throughout all the updates and be able to sample the non-isolated vertices uniformly at random. The latter is exactly the problem solved by $\ell_0$-sampling in streaming algorithms. However, all such algorithms, while only using $O(\poly\log n)$ space, require time $\Omega(\log n)$ per operation. We give a relatively simply data structure that allows to subsample all non-zero entries in a dynamically changing vector of size $n$ in constant time (no matter how small their number might be), albeit with space $O(n)$. We believe that our data
structure might be of independent interest. 

To design a {\em deterministic worst-case} dynamic algorithm we cannot simply invoke the static constant-time algorithm: this algorithm is inherently randomized as it is designed with the goal of reading the smallest possible portion of the graph.
Instead we carefully implement the random local exploration that underlies the static randomized algorithm in a deterministic way. 
Our key observations are {\em (1) that we only need to count the number of CCs that are small in size, i.e.~consist of up to $1/\varepsilon'$ vertices,
	as the number of larger CCs is at most $\varepsilon' \cdot \nis(G)$ {\em and} (2) that these counts can be maintained in  worst-case
	time $O(1/\varepsilon'^2)$ after each update by exploring a neighborhood of $O(1/\varepsilon)$ vertices ``around'' the endpoints of the updated edge.}

Both the randomized and the deterministic MSF algorithm run their respective CC estimation algorithms on each of the $O({\log W}/{\varepsilon})$ relevant subgraphs with $\varepsilon' = \varepsilon/(4W)$. 
Using 
the above-mentioned formula results in
an additive error of $\varepsilon \nis(G)/4$ for MSF.
As the weight of any MSF is at least $\nis(G)/2$, this additive error is at most $\varepsilon M/2$, i.e.,~a
$(1+\varepsilon)$-approximation of $M$. 
For our deterministic algorithm for MSF, the time per edge update is $O(1/\varepsilon'^2) = O(W^2/\varepsilon^2)$ for each of  the $O({\log W}/{\varepsilon})$ subgraphs,  resulting in a 
worst-case $O(W^2 \log W/\varepsilon^3)$ update time. The running time of our randomized algorithm for MSF can be analyzed analogously. %

\subsection{A Deterministic Dynamic Algorithm}\label{sec:estimator_deterministic} 
We first present a deterministic dynamic algorithm for approximating the number of connected components (CCs) with appropriate additive error. %
We use $\ncc(G)$ to denote the number of CCs of $G$, $\nis(G)$ to denote the number of \emph{non-isolated} vertices of $G$, and {\em size} of a CC  to denote the number of vertices in the CC. 

\begin{theorem}\label{thm:deterministic_NCC}
	Let $\varepsilon>0$. There exists a fully dynamic and deterministic algorithm that preprocesses a potentially non-empty graph in $O(\frac{n}{\varepsilon})$ time, and maintains an estimator $\overline{c}$ s.t., $|\overline{c}-\ncc(G)|\leq \varepsilon \cdot \nis(G)$ with worst-case $O(1/\varepsilon^2)$ update time per operation. 
\end{theorem}

We remark that in the above theorem, the initial graph can be an arbitrary graph, and the performance guarantee holds even if the algorithm is not aware of the value $\nis(G)$.  %
By combining the algorithm from Theorem~\ref{thm:deterministic_NCC} and the relation in Lemma~\ref{lem:mst_ncc}, we can obtain the following result.
\begin{theorem}~\label{thm:deterministic_WMST}
	There exists a fully dynamic and deterministic algorithm that maintains an estimator $\overline{M}$ that $(1+\varepsilon)$-approximates the weight $M$ of a MSF of a graph with edge weights from $[1,W]$. The worst-case time per update operation is $O(\frac{W^2\cdot \log W}{\varepsilon^3})$. 
\end{theorem}

\begin{proof}
	Recall that $\nis(G)$ is the number of non-isolated vertices in $G$ and note that $\nis(G)\ge \nis(G^{(\ell_i)})$, since $G^{(\ell_i)}$ is a subgraph of $G$ for any $1 \le i \le r$.
	We call the dynamic algorithm from Theorem~\ref{thm:deterministic_NCC} for estimating $c^{(\ell_i)}$ for each $1 \leq i\leq r$ with $\varepsilon' = 
	\varepsilon /(4W)$, which gives an additive error $\varepsilon\cdot \nis(G^{(\ell_i)})/{4W} \le \varepsilon\cdot \nis(G)/{4W}$. Its worst-case time per update operation 
	is $O(1/\varepsilon'^2)$, which is $O({W^2}/{\varepsilon^2})$.
	Since $G$ is simple, we know that $M\geq {\nis(G)}/{2}$, as each non-isolated vertex is incident
	to at least one edge (of weight at least $1$) of any MSF.
	
	Let $\overline{c}_i$ denote the estimator for $c^{(\ell_i)}$. Then we define $\overline{M}:=n-\overline{c}_r\cdot (1+\frac{\varepsilon}{2})^r+\sum_{i=0}^{r-1}\lambda_i\cdot \overline{c}_i.$ 
	Since $|\overline{c}_i-c^{(\ell_i)}|\leq \frac{\varepsilon \cdot \nis(G)}{4 W}$ and $M\geq \frac{\nis(G)}{2}$, we have that $|\overline{M}-X|\leq \frac{\varepsilon \cdot \nis(G)}{4}\leq \frac{\varepsilon M}{2},$ where $X$ is the quantity in Lemma~\ref{lem:mst_ncc}. Together with inequality~(\ref{eqn:app_mst}), $\overline{M}$ is a $(1+\varepsilon)$-approximation of $M$. %
	Note that the worst-case time per update operation of the algorithm for maintaining $\overline{M}$ is $\sum_{i=1}^r O(\frac{W^2}{\varepsilon^2})=O(\frac{r}{ \varepsilon^2}W^2)=O(\frac{W^2\cdot \log W}{\varepsilon^3})$.
\end{proof}

In the following, we give the proof of Theorem~\ref{thm:deterministic_NCC}.

\begin{proof}[\textbf{Proof of Theorem \ref{thm:deterministic_NCC}}]
	We first give a \emph{static} algorithm for computing the number of small CCs of any graph $G$. We maintain a set of $\lceil 1/\varepsilon \rceil$ counters  $\cnt_\ell$, where $\cnt_{\ell}$ denotes the number of CCs of size $\ell$. 
	 Initially, all the counters are set to  $0$ and all vertices are marked {\em unvisited}. We recursively choose an arbitrary unvisited vertex $v$, mark it  as {\em visited}  and start a BFS at $v$ which runs until (1) it has reached (e.g.~discovered an edge to) $1/\varepsilon +1$ unvisited vertices, 
	 (2) it reaches a visited vertex, or (3) the  BFS terminates because whole connected component (of size at most $1/\varepsilon$) containing $v$ has been explored. 
	 Then we mark all the newly discovered vertices as {\em  visited} and update the counters accordingly. 
	 More precisely, the static and the dynamic algorithms are  as follows.
	\begin{center}
		\begin{tabular}{|p{1.02\textwidth}|}
			\hline
			\textbf{A static algorithm for computing the number of CCs of size at most $1/\varepsilon$}\\
			\begin{enumerate}
				\item Initialize $\cnt_\ell=0$, for each $1\leq \ell \leq 1/\varepsilon$. Mark all vertices as unvisited.
				\item While there exists some unvisited vertex $v$: 
				\begin{enumerate}
					\item Do BFS from $v$ until (i) $1/\varepsilon+1$ unvisited vertices have been reached, or (ii) any visited vertex has been reached, or (iii) no more new vertices can be reached. Mark all the newly discovered vertices in the search as visited.
					\item If (iii) occurs, and $\ell$ vertices have been reached for some $\ell \leq \frac{1}{\varepsilon}$, then increment $\cnt_\ell$ by $1$.
				\end{enumerate}
				\item Define the estimator  $\overline{c}:=\sum_{\ell=1}^{1/\varepsilon} \cnt_\ell$.
			\end{enumerate}
					\vspace{-0.5em}
			\\ \hline
		\end{tabular}
	\end{center}

	The dynamic algorithm %
	updates the counter $\overline{c}$ in time $O(1/\varepsilon^2)$ by running a limited BFS from $u$ and $v$ in the graph before and after the update. The details are given below.%
	\begin{figure}[h]
	\begin{center}
		\begin{tabular}{|p{\textwidth}|}
			\hline
			\textbf{Maintaining an estimator for $\ncc(G)$ of a dynamic graph $G$}\\
			\begin{enumerate}	
				\item Preprocessing: run the above static algorithm to find the $\overline{c}$, the number of CCs of $G_0$ of size at most $1/\varepsilon$.

				\item Handling an edge insertion $(u,v)$: perform three BFS calls: two from $u$ and $v$, respectively, in the graph before the insertion of $(u,v)$, and one from $u$ in the graph after 
				the insertion. 
				Stop the BFS once $1/\varepsilon+1$ vertices have been reached or no more new vertices can be reached. 
				Let $s_u^{(0)},s_v^{(0)},s_u^{(1)}$ denote the sizes of the corresponding explored subgraphs.
				 
				\begin{enumerate}
					\item\label{alg:one_large} If exactly one of $s_u^{(0)}$ and $s_v^{(0)}$, say $s_u^{(0)}$, is no larger than $1/\varepsilon$, then decrement $\overline{c}$ by $1$.
					
					\item\label{alg:two_small} If both of $s_u^{(0)},s_v^{(0)}$ are smaller than $1/\varepsilon$:
					\begin{enumerate}
				\item\label{alg:two_small_1}  $s_u^{(1)}$ is larger than $1/\varepsilon$, then decrement $\overline{c}$ by $2$; 
				\item\label{alg:two_small_2} $s_u^{(1)}$ is no larger than $1/\varepsilon$ and $s_u^{(1)}=s_u^{(0)}$, then keep $\overline{c}$ unchanged;
				
				\item\label{alg:two_small_3} $s_u^{(1)}$ is no larger than $1/\varepsilon$ and $s_u^{(1)}\neq s_u^{(0)}$, then decrement $\overline{c}$ by $1$.
					\end{enumerate}
				\end{enumerate} %
				\item Handling an edge deletion $(u,v)$: perform three BFS calls: one from $u$ in the graph before the deletion of $(u,v)$, and two from $u$ and $v$, respectively, in the graph after 
				the deletion. 
				Stop the BFS once $1/\varepsilon+1$ vertices have been reached or no more new vertices can be reached. 
				Let $s_u^{(0)},s_u^{(1)},s_v^{(1)}$ denote the sizes of the corresponding explored subgraphs.
				
				\begin{enumerate}
					\item If exactly one of $s_u^{(1)}$ and $s_v^{(1)}$, say $s_u^{(1)}$, is no larger than $1/\varepsilon$, then increment $\overline{c}$ by $1$.
					
					\item If both of $s_u^{(1)},s_v^{(1)}$ are smaller than $1/\varepsilon$:
					\begin{enumerate}
						\item  $s_u^{(0)}$ is larger than $1/\varepsilon$, then increment $\overline{c}$ by $2$; 
						\item $s_u^{(0)}$ is no larger than $1/\varepsilon$ and $s_u^{(0)}=s_u^{(0)}$, then keep $\overline{c}$ unchanged;
						
						\item $s_u^{(0)}$ is no larger than $1/\varepsilon$ and $s_u^{(0)}\neq s_u^{(1)}$, then increment $\overline{c}$ by $1$.
					\end{enumerate}
				\end{enumerate}
			\end{enumerate}		
			\\ \hline
		\end{tabular}
	\end{center}
	\end{figure}

	{\bf Correctness.} For the correctness of the dynamic algorithm, we let $\nscc(G)$ denote the number of CCs of size at most $1/\varepsilon$ in $G$. We show that the maintained estimator $\overline{c}$ is equal to $\nscc(G)$ throughout all the updates. Note that we preprocess the graph using the above static algorithm and obtain the estimator $\overline{c}$ for the initial graph. 
	By definition, $\overline{c}=\nscc(G_0)$. Now for any edge insertion $(u,v)$, we know that the number $\nscc$ (of CCs of size at most $1/\varepsilon$) can change by at most $2$. More precisely, it changes if and only if at least one of $s_u^{(0)},s_v^{(0)}$ is at most $1/\varepsilon$ and $u,v$ do not belong to the same CC before the edge insertion, where $s_u^{(0)}$ and $s_v^{(0)}$ are the sizes of the explored subgraphs (before the edge insertion) starting from $u$ and $v$, respectively, that we compute in the algorithm. 
	Furthermore, if Step~\ref{alg:one_large} happens, i.e., exactly one of $s_u^{(0)}$ and $s_v^{(0)}$, say $s_u^{(0)}$, is no larger than $1/\varepsilon$, then a small CC merges into a large CC, and thus $\nscc$ decreases by $1$. If Step~\ref{alg:two_small} happens (i.e., $s_u^{(0)},s_v^{(0)}$ are smaller than $1/\varepsilon$): 
	if Step~\ref{alg:two_small_1} happens, i.e., $s_u^{(1)}$ is larger than $1/\varepsilon$, then two small CCs merge into a CC of size larger than $1/\varepsilon$ and thus $\nscc$ decreases by $2$; if Step~\ref{alg:two_small_2} happens, i.e., $s_u^{(1)}$ is no larger than $1/\varepsilon$ and $s_u^{(1)}=s_u^{(0)}$, then $u,v$ belong to the same CC before $(u,v)$ was inserted and thus $\nscc$ remains unchanged; if Step~\ref{alg:two_small_3} happens, i.e., $s_u^{(1)}$ is no larger than $1/\varepsilon$ and $s_u^{(1)}\neq s_u^{(0)}$, then two small CCs merge into a CC of size no larger than $1/\varepsilon$ and thus $\nscc$ decreases by $1$.
By the description of our algorithm, after the insertion $(u,v)$, the maintained $\overline{c}$ still satisfies that $\overline{c}=\nscc(G')$, where $G'$ is the updated graph. The case for edge deletions can be analyzed similarly.

	Since the total number of CCs of size larger than $\frac{1}{\varepsilon}$ is at most $\varepsilon \cdot \nis(G)$, where $\nis(G)$ is the number of non-isolated vertices of $G$, we know that $\overline{c}$ approximates $\ncc(G)$ with an additive error $\varepsilon \cdot \nis(G)$.

{\bf Running time.} Now we analyze the running time of our dynamic algorithm. We first show that our static algorithm for preprocessing the initial graph can be implemented in $O(n\cdot \frac{1}{\varepsilon})$ time. %
Note that it suffices to bound the time of {\em exploring } each CC $C$, i.e., until all the vertices inside $C$ have been marked as visited. Note that $\cnt_\ell$ is exactly the number of CCs of size $\ell$, for $\ell\leq 1/\varepsilon$ and consider two cases, which together show the $O(n/\varepsilon)$ bound.
(1) If $|C|=\ell \leq \frac{1}{\varepsilon}$, then the total time for exploring $C$ is $O(\ell^2)$. In this case, we note that the total time for exploring CCs of size at most $1/\varepsilon$ is $\sum_{\ell=1}^{1/\varepsilon} \cnt_\ell\cdot O(\ell^2)\leq \sum_{\ell=1}^{1/\varepsilon} \cnt_\ell\cdot \ell \cdot O(1/\varepsilon)= O(n/\varepsilon)$, where the last equation follows from the fact that $\sum_{\ell=1}^{1/\varepsilon}\cnt_\ell\cdot \ell\leq n$. 

(2) If $|C|>1/\varepsilon$, let $S=\{v_1,v_2,\cdots, v_b\}$ denote the set of vertices from which we start a BFS in $C$ and let $s_i$ denote the number of newly discovered vertices from vertex $v_i$. 
It holds that $s_i\leq 1/\varepsilon+1$ for each $i\leq b$ by the description of our algorithm. 
Let $t_j$ denote the number of vertices in $S$ from which the BFS discovers exactly $j$ new vertices, for each $j\leq 1/\varepsilon+1$. Then $|C|=\sum_{j=1}^{1/\varepsilon+1} t_j\cdot j$. 
Furthermore, we note that for each $j\geq 1$, it takes time $O(j\cdot \frac{1}{\varepsilon})$ for the BFS to discover exactly $j$ new vertices, as we will only scan at most $\frac{1}{\varepsilon}+1$ neighbors for each of these new vertices. 
Thus, the total time of exploring $C$ is  $\sum_{j=1}^{1/\varepsilon+1} t_j\cdot O(j\cdot \frac{1}{\varepsilon})\leq O(1/\varepsilon) \cdot \sum_{j=1}^{1/\varepsilon+1} t_j\cdot j =O(|C|/\varepsilon)$. 
Thus, the total time of exploring CCs of size at least $1/\varepsilon+1$ is $\sum_{C: |C|\geq 1/\varepsilon+1} O(|C|/\varepsilon)=O(n/\varepsilon)$, where the last equation follows from the fact that $\sum_{C:|C|\geq 1/\varepsilon + 1} |C| \leq n$.  %
	Finally, we note that for each update (either insertion or deletion), we only need to execute $O(1)$ BFS calls, each of which will explore at most $O(1/\varepsilon)$ vertices (and thus $O(1/\varepsilon^2)$ edges). Therefore, the worst-case time per update operation is $O(1/\varepsilon^2)$.
\end{proof}

\subsection{A Randomized Dynamic Algorithm}\label{subsec:randomized}
In this section, we give a randomized dynamic algorithm for estimating the weight of the MSF. Our algorithm will be built upon a dynamic algorithm for approximating $\ncc(G)$ with an additive error $\varepsilon \cdot \Thr(G)$, for some parameter $\Thr(G)\geq \nis(G)$. We have the following result.

\begin{theorem}\label{thm:random_NCC}
	Let $1>\varepsilon' >0$ and $0<p<1$. Let $G$ be a dynamically changing graph such that each update operation has as additional parameter a value $\Thr(G)$ such that (a) $\Thr(G)\geq \nis(G)$ (where $G$ denotes the graph right before the update) and (b) each update changes $\Thr(G)$ by at most $2$ in comparison to the previous update. 
	Then there exists a fully dynamic algorithm that takes as input the initial graph and the sequence of update operations and, with probability at least $1-p$, maintains an estimator $\overline{cc}$ for the number $\ncc$ of CCs of a graph $G$ s.t.,  $|\overline{cc}-\ncc(G)|\leq \varepsilon'\cdot \Thr(G)$. The worst-case time per update operation is $O(\max\{1,\frac{\log(1/\varepsilon')\log(1/p)}{\varepsilon'^3 \Thr^{*}}\})$, where $\Thr^{*}$ is the minimum value of $\Thr(G)$ over all updates. Our algorithm works against an adaptive adversary.
\end{theorem}

We defer the proof of the above theorem to Section~\ref{subsec:proof_random_ncc}. Given Theorem~\ref{thm:random_NCC} and the relation from Lemma~\ref{lem:mst_ncc}, we have the following theorem.

\begin{theorem}~\label{thm:random_WMST}
	Let $0<p'<1$. There exists a fully dynamic algorithm that with probability at least $1-p'$, maintains an estimator $\overline{M}$ that is a $(1+\varepsilon)$-approximation of the weight $M$ of MSF of a graph $G$ with edge weights from $[W]$. The worst-case time per update operation is $O(\max\{1,\frac{W^3\log W\log(\frac{W}{\varepsilon} )\log(\frac{\log W}{\varepsilon p'})}{\varepsilon^4 \nis^*}   \})$, where $\nis^*$ is the minimum number of non-isolated vertices in $G$ throughout all the updates. Our algorithm works against an adaptive adversary.
\end{theorem}
\begin{proof}

	Recall from the proof of Theorem \ref{thm:deterministic_WMST}, $\nis(G)\ge \nis(G^{(\ell_i)})$, since $G^{(\ell_i)}$ is a subgraph of $G$ for any $1 \le i \le r$. 
Since $G$ is simple, we know that $M\geq {\nis(G)}/{2}$, as each non-isolated vertex is incident
	to at least one edge (of weight at least $1$) of any MSF.

Now for each $j\leq r$, we would like to maintain $c^{(\ell_j)}$, the number of CCs in $G^{(\ell_j)}$, by invoking Theorem~\ref{thm:random_NCC}. In order to do so, we first ensure that $G^{(\ell_j)}$ will update with $G$ ``synchronously'': for each edge update $(u,v)$ in $G$, if the weight of $(u,v)$ is at most $\ell_j$, then we update $G^{(\ell_j)}$ accordingly; if the weight of $(u,v)$ is larger than $\ell_j$, then we update $G^{(\ell_j)}$ by first inserting a self-loop $(u,u)$ and then immediately deleting the self-loop $(u,u)$. In the latter case, each update in $G$ corresponds to two updates in $G^{(\ell_j)}$, which guarantee that $G^{(\ell_j)}$ is unchanged after the updates. 

Now for each $j\leq r$, we execute the algorithm
	of Theorem~\ref{thm:random_NCC} on
	$G^{(\ell_j)}$ (which is updated according to the above scheme) using $p=\frac{p' }{r}$, $\varepsilon' = \frac{\varepsilon}{4 W}$, and $\Thr(G^{(\ell_j)})=\nis(G)$, i.e.,
	each update operation uses as additional parameter $\nis(G)$.
Note that it always holds that $\Thr(G^{(\ell_j)})\geq \nis(G^{\ell_j})$, and each update in $G^{(\ell_j)}$ changes $\Thr(G^{(\ell_j)})$ by at most $2$ in comparison to the previous update, which is guaranteed by the above update sequence.
Thus, by~Theorem~\ref{thm:random_NCC} the algorithm computes an estimator $\overline{c}_j$ for $c^{(\ell_j)}$ such that with probability $1-\frac{p' }{r}$, it holds that 
	\begin{eqnarray*}
		\abs{\overline{c}_j-c^{(\ell_j)}}\leq \varepsilon'\cdot \Thr(G^{(\ell_j)}) =\frac{\varepsilon}{4 W} \cdot \nis(G), %
	\end{eqnarray*}
Note that throughout all the updates, it holds that $T(G^{(\ell_j)})=\nis(G)\geq \nis^*$. Thus the amortized time spent per update for computing $\overline{c}_j$ is
	\begin{eqnarray*}
		O(\max\{1,\frac{\log(1/\varepsilon')\log(1/p)}{\varepsilon'^3 \nis^{*}}\})&=&O(\max\{1, \frac{W^3\log(W/\varepsilon)\log(r/p')}{\varepsilon^3 \nis^*} \})%
	\end{eqnarray*} %
	
	Let $\overline{M}=n-\overline{c}_r\cdot (1+\frac{\varepsilon}{2})^r+\sum_{j=0}^{r-1}\lambda_j\cdot \overline{c}_j$. Since $|\overline{c}_j-c^{(\ell_j)}|\leq \frac{\varepsilon \nis(G)}{4W}$ and $M\geq \frac{\nis(G)}{2}$, we have that $|\overline{M}-X|\leq \frac{\varepsilon \nis(G)}{4}\leq \frac{\varepsilon M}{2}$, where $X$ is as defined in Lemma~\ref{lem:mst_ncc}. Together with inequality~(\ref{eqn:app_mst}), $\overline{M}$ is a $(1+\varepsilon)$-approximation of $M$. The success probability of the algorithm is at least $1-r\cdot \frac{p'}{r}=1-p'$, and the worst-case time per update operation is 
	\begin{eqnarray*}
		\max\{O(1),\sum_{j=1}^{r} O( \frac{W^3\log(W/\varepsilon )\log(r/p')}{\varepsilon^3 \nis^*} )\} =O(\max\{1,\frac{W^3\log W\log(\frac{W}{\varepsilon} )\log(\frac{\log W}{\varepsilon p'})}{\varepsilon^4 \nis^*}  \})
	\end{eqnarray*}

	The algorithm works against an adaptive adversary as each of the algorithms from ~Theorem~\ref{thm:random_NCC} works against an adaptive adversary and the MSF algorithm
	simply computes a weighted sum of the values returned by each of these algorithms.
	This completes the proof of the theorem.
\end{proof}
The following is a direct corollary of the above theorem by setting $p'=1/n^c$ and the fact that $\nis^*\leq n$. %
\begin{corollary}
	Let $c$ be any constant such that $c\geq 1$. There exists a fully dynamic algorithm that with probability at least $1-\frac{1}{n^c}$, maintains an estimator $\overline{M}$ that is a $(1+\varepsilon)$-approximation of the weight $M$ of the MSF of a graph $G$ with edge weights from $[W]$ and	$W=O((\nis^*)^{1/3}/\log^3 n)$, where $\nis^*$ is the minimum number of non-isolated vertices in $G$ throughout all the updates. The algorithm runs in  $O(\frac{1}{\varepsilon^4}\log^2(\frac{1}{\varepsilon}))$ worst-case time per update operation. 
\end{corollary}

We note that $\nis^*\geq 2 m^*$, where $m^*$ is the minimum number of edges of the graph throughout all the updates. This is true as for the graph $G$ with minimum non-isolated vertices, i.e., $\nis(G)=\nis^*$, each non-isolated vertex will contribute at least half of an edge, and thus the number of edges in $G$ is at least $\frac{\nis^*}{2}$, which is at least $m^*$ by the definition of $m^*$. Then we have the following corollary.

\begin{corollary}
	Let $c$ be any constant such that $c\geq 1$. There exists a fully dynamic algorithm that with probability at least $1-\frac{1}{n^c}$, maintains an estimator $\overline{M}$ that is a $(1+\varepsilon)$-approximation of the weight $M$ of the MSF of a graph $G$ with edge weights from $[W]$ and	$W= O((m^*)^{1/3}/\log^3 n)$, where $m^*$ is the minimum number of edges of the graph throughout all the updates. The algorithm runs in  $O(\frac{1}{\varepsilon^4}\log^2(\frac{1}{\varepsilon}))$ worst-case time per update operation. 
\end{corollary}

\subsubsection{Proof of Theorem~\ref{thm:random_NCC}}\label{subsec:proof_random_ncc}
We first state a known constant-time static algorithm for estimating the number $\ncc(G)$ of CCs with additive error $\varepsilon n$, building on which, we then give a dynamic algorithm for estimating $\ncc(G)$ with an additive error $\varepsilon \Thr(G)$, for some parameter $\Thr(G)\geq \nis(G)$. 

\paragraph{Static algorithms for estimating the number of CCs.}
Recall that $\ncc(G)$ denotes the number of CCs of a graph $G$. We need the following lemma by Berenbrink et al.~\cite{BKM14:numcc} (which improves upon the result in~\cite{CRT05:MST}) that gives a constant-time algorithm for estimating $\ncc(G)$. It is assumed that the algorithm can make some types of queries\footnote{Please note that the query access to the input graph from the sublinear-time community is different from the query operation in the dynamic algorithms community. } to access to the graph. That is, the algorithm can perform a \textbf{vertex-sample query}, which allows it to sample a vertex uniformly at random from $V$, and can make queries to the adjacency list of the graph. Note that these two queries for accessing a static graph can be supported by maintaining an array of vertices and the adjacency list of the graph, respectively.%
\begin{lemma}[\cite{BKM14:numcc}]\label{lemma:bkm_ncc}
	Let $\varepsilon >0$ and $0<p<1$. Suppose the algorithm has access to the adjacency list of a graph $G$ and can perform vertex-sample queries. Then there exists an algorithm that with probability at least $1-p$, returns an estimate that approximates $\ncc(G)$ with an additive error $\varepsilon n$. The running time of the algorithm is $O(1/\varepsilon^2\log(1/\varepsilon)\log(1/p))$. 
\end{lemma}

We remark that the algorithm in~\cite{BKM14:numcc} simply samples (uniformly at random) $O(1/\varepsilon^2)$ vertices,
performs a BFS starting from each sampled vertex (for a number of steps) and then makes decisions based on the explored subgraphs. %
Note that if the algorithm is able to perform a \textbf{non-isolated vertex-sample} query, i.e., the algorithm can sample a vertex uniformly at random from the set $N$ of all non-isolated vertices in a graph $G$, then one can approximate the size of CCs in the subgraph $G[N]$ induced by vertices in $N$ with an additive error $\varepsilon |N|=\varepsilon \nis(G)$. This is true as we can simply treat $G[N]$ as the input graph in the algorithm from Lemma~\ref{lemma:bkm_ncc}. We let $G_{\nis}=G[N]$, and thus $\ncc(G_{\nis})$ denotes the number of CCs in $G[N]$. We have the following corollary.
\begin{corollary}\label{lemma:bkm_ncc_nis}
	Let $\varepsilon >0$ and $0<p<1$. Let $\nis(G)$ be the number of non-isolated vertices in $G$. Suppose the algorithm has access to the adjacency list of a graph $G$ and can  perform non-isolated vertex-sample queries. Then there exists an algorithm that with probability at least $1-p$, returns an estimate $\overline{b}$ that approximates $\ncc(G_{\nis})$ with an additive error $\varepsilon\cdot \nis(G)$. The running time of the algorithm is $O(1/\varepsilon^2\log(1/\varepsilon)\log(1/p))$. 
\end{corollary}

\paragraph{Estimating $\ncc(G)$: from static to dynamic.} 
In order to dynamically maintaining an estimate for $\ncc(G)$ with an additive error $\varepsilon \Thr(G)$ for some $\Thr(G)\geq \nis(G)$, we will periodically invoke the algorithm from Corollary~\ref{lemma:bkm_ncc_nis} as a subroutine for our dynamic algorithm. This requires us to maintain some data structures so that the algorithm can query the adjacency list of the graph and perform non-isolated vertex-sample queries at any time. The adjacency list of a dynamic graph can be updated trivially in constant time. Next we give a data structure to support non-isolated vertex-sample queries in a dynamic setting.

\setlist[enumerate]{itemsep=0pt,topsep=1pt}
\paragraph{Data structure for supporting non-isolated vertex-sample queries.} 
We first present a more general data structure to sample non-zero entries from an array and then show how to use it to support non-isolated vertex-sample queries.

Given a set $V$ of  $n$ elements (here vertices), numbered from $0$ to $n-1$, each element $u$ with an associated number $d_u$ (here degree), we show how to support the following operations in constant time with preprocessing time $O(n)$:
\begin{center}
	\begin{tabular}{|p{1.02\textwidth}|}
		\hline
\begin{enumerate}
	\item[--]{\bf Update($u,\delta$):} add $\delta$ to $d_u$, where $\delta$ can be positive or negative.
	\item[--]{\bf Non-zero sample():} return an element that is chosen uniformly at random from all elements $u$ with $d_u \not = 0$.
\end{enumerate}
		\vspace{-1em}
		\\ \hline
	\end{tabular}
\end{center}

Let us call an element $u$ of $V$ with $d_u \not= 0$ a {\em non-zero} element.
We implement the data structure by using two arrays and a counter:
\begin{enumerate}
	\item We keep the number $\nis$ of non-zero elements of $V$.
	\item We keep an array $\mathcal{A}$ of size $n$, where only the first $\nis$  entries are used, such that (i) each entry in $\mathcal{A}$ stores a non-zero element $u$  together with $d_u$ and (ii) each non-zero element of $V$ is stored in $\mathcal{A}$ within the first $\nis$ entries.
	\item  We keep an array $\mathcal{P}$ of size $n$, which has an entry for every element of $V$, such that if an element $u$ is stored in $\mathcal{A}[i]$ (i.e. $u$ is non-zero), then $\mathcal{P}[u] = i$; and if an element $u$ is not stored in $\mathcal{A}$ (i.e. $d_u = 0$), then $\mathcal{P}[u] = -1$. Thus $\mathcal{P}$ consists of indices corresponding to the positions of elements in $\mathcal{A}$ or the number $-1$.
\end{enumerate}
During preprocessing we initialize both arrays, set all entries of $\mathcal{P}$ to -1, and set $\nis$ to 0. Then we insert every element $u$ whose initial value $d_u \not= 0$ by calling {Update}($u, d_u$).

\paragraph{Handling an \textbf{Update}($u, \delta$) operation.} Whenever an {Update}($u, \delta$) operation is executed, we check if $\mathcal{P}[u] > -1$. 

Case (1): $\mathcal{P}[u] > -1$. This means that $u$ is stored in $\mathcal{A}$ and $\mathcal{P}[u]$ contains the index of $u$ in $\mathcal{A}$.
Thus we add $\delta$ to $d_u$, which is retrieved and then stored in the entry $\mathcal{A}[\mathcal{P}[u]]$. If the resulting value $d_u \not= 0$, this completes the update operation.
If, however, the resulting value $d_u = 0$, let $v$ be the element stored in $\mathcal{A}[\nis]$.
We copy into
$\mathcal{A}[\mathcal{P}[u]]$ all information of element $v$, which we retrieve from $\mathcal{A}[\nis]$. Then we set
$\mathcal{P}[v] = \mathcal{P}[u]$, set $\mathcal{P}[u] = -1$, and decrement $\nis$.

Case (2): $\mathcal{P}[u] =-1$.
We increment $\nis$ by 1, set $d_u = \delta$, store $u$ and $d_u$ in $\mathcal{A}[\nis]$, and set $\mathcal{P}[u] = \nis$.

\paragraph{Handling \textbf{Non-zero sample} operation.} To implement a {Non-zero sample} operation, we pick a random integer number $j$ between $0$ and $\nis-1$ and return the element from $\mathcal{A}[j]$.

\paragraph{Supporting non-isolated vertex-sample queries in dynamic graphs.} Next we show how to use the above data structure to support non-isolated vertex-sample query throughout all the updates.
Whenever an edge $(u,v)$ is inserted, for each $x \in \{u,v\}$, we call Update($x, 1$).
Whenever an edge $(u,v)$ is deleted, for each $x \in \{u,v\}$, we call Update($x, -1$).
To sample a non-isolated vertex, we call Non-zero sample().

\textbf{Remark:} It is interesting to contrast our data structure for non-isolated vertex-sample queries to the sketches for $\ell_0$-sampling in the data streaming community. To the best of our knowledge, all the sketches for $\ell_0$-sampling use only $O(\poly\log n)$ space, but
require $\Omega(\log n)$ update time, while we use $O(n)$ space, but require only constant time.

\paragraph{Dynamically estimating $\ncc(G)$ with an additive error $\varepsilon \cdot \Thr(G)$, for a parameter $\Thr(G)\geq \nis(G)$.}
Now we are ready to describe our randomized dynamic algorithm for estimating $\ncc(G)$ with an additive error $\varepsilon \cdot \Thr(G)$. 
Our idea is as follows: %
We will maintain the value $\Gamma=\nis(G)$. During initialization, we set  $\Psi=\Thr(G)$ and $\overline{c}=\ncc(G)$. Then we repeat the following: after every $\varepsilon' \Psi/4$ updates, we re-compute the estimator $\overline{c}$ by invoking the static algorithm~from Corollary~\ref{lemma:bkm_ncc_nis} on the current graph $G$ with parameter $\varepsilon'/4$ and re-set  $\Psi=\Thr(G)$. 
In the meanwhile, we maintain the adjacency list of the dynamic graph in a trivial way and maintain the data structures for supporting non-isolated vertex-sample queries as above. The description of our randomized algorithm is given as follows. (For simplicity, we did not include the details for maintaining adjacency list, array of degrees, and data structures for sampling non-isolated vertices.)
\begin{center}
	\begin{tabular}{|p{\textwidth}|}
		\hline
		\textbf{Maintaining an estimator $\overline{c}$ for $\ncc(G)$ of a dynamic graph $G$ with additive error $\varepsilon'\cdot \Thr(G)$, for some parameter $\Thr(G)\geq \nis(G)$}\\
		\begin{enumerate}	
			\item Preprocessing: Traverse the initial graph $G$ (e.g., by performing BFS) to obtain $\nis(G)$ and $\ncc(G)$. Start of the first
			phase. Initialize $\Gamma=\nis(G)$ and $\overline{c}=\ncc(G)$, $\Psi=\Thr(G)$. Let $i=1$. 

			\item For the $i$-th update: %
			\begin{enumerate}
				\item update $\Gamma$ to be $\nis(G)$
				
				\item if $i \mod (\frac{\varepsilon'\cdot\Psi}{4}) = 0$, then \hspace{2cm} $\triangleright$ New phase starts
				\begin{enumerate}
					\item compute an estimator $\overline{b}$ for $\ncc(G_{\nis})$ by running the static algorithm in Corollary~\ref{lemma:bkm_ncc_nis} on $G$ with parameter $\varepsilon=\frac{\varepsilon'}{4}$
					\item set $\overline{c}=\overline{b}+n-\Gamma$ \hspace{2cm} $\triangleright$ $n - \Gamma$ is the number of isolated nodes in $G$ %
					\item set $\Psi=\Thr(G)$
				\end{enumerate}
				\item set $i=i+1$
			\end{enumerate} %
			
		\end{enumerate}		
		\\ \hline
	\end{tabular}
\end{center}

Note that the algorithm runs a static BFS traversal to obtain the exact values for $\nis(G)$ and $\ncc(G)$ in the initial graph.
Then it
partitions the updates into \emph{phases}, starting a new
phase every $\varepsilon' \Psi/4$ updates. At the beginning of each phase,  it sets  $\Psi=\Thr(G_i)$, where $i$ is the update at the beginning of the
phase, and runs the constant-time algorithm from~Corollary~\ref{lemma:bkm_ncc_nis}). (The parameter $\Thr(G_i)$ is ignored for all updates that do not happen at the beginning of a phase.) When asked a query,
it returns $\overline{c}$.

{\bf Correctness.}
We consider an arbitrary phase. At the beginning of the phase either
the algorithm computes the correct value of $\overline{c}$ (for the first phase) or it calls the the static algorithm, which returns with probability $1-p$
an estimator $\overline{b}$ for $\ncc(G_{\nis})$
such that $|\overline{b}-\ncc(G_{\nis})|\leq \frac{\varepsilon' \nis(G)}{4}$, which gives 
$$|\overline{c}-\ncc(G)|=|\overline{b}+n-\nis(G)-\ncc(G)|=|\overline{b}-\ncc(G_{\nis})|\leq \frac{\varepsilon' \nis(G)}{4},$$
where the second equation follows from the fact that $\ncc(G)$ is the sum of $\ncc(G_{\nis})$ and the number of isolated vertices, $n-\nis(G)$.
Let $\Lambda = \nis(G)$ at the beginning of the phase and note that
$\Gamma$ always equals the value $\Thr(G)$ that was given by the first update of a phase. 
We are guaranteed that
at each update $\Thr(G) \ge \nis(G)$ and, thus, it follows that $\Gamma \ge \Lambda$.

We analyze the additive error throughout the phase, ie. the next $\frac{\varepsilon' \Psi}{4}$ updates.
As each update changes $\ncc(G)$ by at most 1,
with at least probability $1-p$, it holds that $\abs{\overline{c}-\ncc(G)}\leq \frac{\varepsilon' \Lambda}{4}+\frac{\varepsilon' \Psi}{4} \leq \frac{\varepsilon' \Psi}{2}$.
 Note that $|\Psi-\Thr(G)|\leq \frac{\varepsilon'\Psi}{2}$, as each update (for all $\frac{\varepsilon' \Psi}{4}$ updates) can change $\Thr(G)$ by at most $2$. Thus $\Psi\leq \frac{1}{1-\varepsilon'/2}\Thr(G)\leq (1+\varepsilon')\Thr(G)\leq 2\Thr(G)$. This implies that $\overline{c}$ approximates $\ncc(G)$ with an additive error $\varepsilon' \Thr(G)$ at any time in a phase. 

Thus it follows that with probability $1-p$, at any time $\abs{\overline{c}-\ncc(G)}\leq \varepsilon' \Thr(G)$.

Note that the algorithm uses ``fresh'' random bits at the beginning of each phase, only needs to access to the current graph, and does not reuse any information computed in prior phases. Within each phase we performed a worst-case analysis, i.e., we
assumed that the adversary changes the graph in the worst possible way, i.e., changing $\ncc(G)$ by 1 in each update.
Thus, our algorithm works against an adaptive adversary, i.e. an adversary that sees the answers to all queries {\em before} deciding on the next update operation.

{\bf Running time.}
For each phase with parameter $\Psi$, the amortized running time is $$O(\max\{1,\frac{(1/\varepsilon^2)\log(1/\varepsilon)\log(1/p)}{\varepsilon \Psi}\})=O(\max\{1,\frac{\log(1/\varepsilon')\log(1/p)}{\varepsilon'^3 \Psi}\}).$$
(Note that we always need to use $O(1)$ time to update the adjacency list and other data structures so as to provide query access to the graph). If we let $\Thr^{*}$ denote the minimum value $\Thr(G)$ over all the graphs throughout all the updates, then in any phase, the parameter $\Psi\geq \Thr^{*}$ and the amortized running time of the algorithm is
$$O(\max\{1,\frac{\log(1/\varepsilon')\log(1/p)}{\varepsilon'^3 \Thr^{*}}\}).$$

By using the standard global rebuilding technique, we can de-amortize the running time and obtain $O(\max\{1,\frac{\log(1/\varepsilon')\log(1/p)}{\varepsilon'^3 \Thr^{*}}\})$ worst-case time per update operation. This finishes the proof of Theorem~\ref{thm:random_NCC}.

\textbf{Remark:} We further remark that by using a similar algorithm and analysis, we can maintain an estimator for $\ncc(G)$ with an additive error $\varepsilon n^{O(1)}$ (instead of $\varepsilon\cdot\Thr(G)$ or $\varepsilon\cdot \nis(G)$), which might be of independent interest. We defer the details to Appendix~\ref{app:4}.

\bibliographystyle{alpha}
\bibliography{dynamic_constant}

\newcommand{\etalchar}[1]{$^{#1}$}
\begin{thebibliography}{HRWN15}

\bibitem[ACK19]{ACK19:coloring}
Sepehr Assadi, Yu~Chen, and Sanjeev Khanna.
\newblock Sublinear algorithms for {($\Delta+ 1$)} vertex coloring.
\newblock In {\em Proceedings of the Thirtieth Annual ACM-SIAM Symposium on
  Discrete Algorithms}, pages 767--786. SIAM, 2019.

\bibitem[AGM12]{AGM12:linear}
Kook~Jin Ahn, Sudipto Guha, and Andrew McGregor.
\newblock Analyzing graph structure via linear measurements.
\newblock In {\em Proceedings of the twenty-third annual ACM-SIAM symposium on
  Discrete Algorithms}, pages 459--467. Society for Industrial and Applied
  Mathematics, 2012.

\bibitem[BCHN18]{BCHN18:coloring}
Sayan Bhattacharya, Deeparnab Chakrabarty, Monika Henzinger, and Danupon
  Nanongkai.
\newblock Dynamic algorithms for graph coloring.
\newblock In {\em Proceedings of the Twenty-Ninth Annual ACM-SIAM Symposium on
  Discrete Algorithms}, pages 1--20. SIAM, 2018.

\bibitem[BCK{\etalchar{+}}17]{barba2017dynamic}
Luis Barba, Jean Cardinal, Matias Korman, Stefan Langerman, Andr{\'e} van
  Renssen, Marcel Roeloffzen, and Sander Verdonschot.
\newblock Dynamic graph coloring.
\newblock In {\em Workshop on Algorithms and Data Structures}, pages 97--108.
  Springer, 2017.

\bibitem[BFP{\etalchar{+}}73]{blum1973time}
Manuel Blum, Robert~W Floyd, Vaughan Pratt, Ronald~L Rivest, and Robert~E
  Tarjan.
\newblock Time bounds for selection.
\newblock {\em Journal of Computer and System Sciences}, 7(4):448--461, 1973.

\bibitem[BGK{\etalchar{+}}]{BComm}
Sayan Bhattacharya, Fabrizio Grandoni, Janardhan Kulkarni, Quanquan~C. Liu, and
  Shay Solomon.
\newblock Fully dynamic {$(\Delta+1)$} coloring in constant update time.

\bibitem[BK19]{BhattacharyaK2019}
Sayan Bhattacharya and Janardhan Kulkarni.
\newblock Deterministically maintaining a $(2+\varepsilon)$-approximate minimum
  vertex cover in ${O}(1/\varepsilon^2)$ amortized update time.
\newblock In {\em Proceedings of the Thirtieth Annual ACM-SIAM Symposium on
  Discrete Algorithms}, pages 1872--1885. SIAM, 2019.

\bibitem[BKMT14]{BKM14:numcc}
Petra Berenbrink, Bruce Krayenhoff, and Frederik Mallmann-Trenn.
\newblock Estimating the number of connected components in sublinear time.
\newblock {\em Information Processing Letters}, 114(11):639--642, 2014.

\bibitem[BKS12]{BKS12:spanner}
Surender Baswana, Sumeet Khurana, and Soumojit Sarkar.
\newblock Fully dynamic randomized algorithms for graph spanners.
\newblock {\em ACM Transactions on Algorithms (TALG)}, 8(4):35, 2012.

\bibitem[BM17]{barenboim2017fully}
Leonid Barenboim and Tzalik Maimon.
\newblock Fully-dynamic graph algorithms with sublinear time inspired by
  distributed computing.
\newblock {\em Procedia Computer Science}, 108:89--98, 2017.

\bibitem[CFPS19]{CFPS19}
Artur Czumaj, Hendrik Fichtenberger, Pan Peng, and Christian Sohler.
\newblock Testable properties in general graphs and random order streaming.
\newblock {\em CoRR}, abs/1905.01644, 2019.

\bibitem[CHHK16]{censor2016optimal}
Keren Censor-Hillel, Elad Haramaty, and Zohar Karnin.
\newblock Optimal dynamic distributed mis.
\newblock In {\em Proceedings of the 2016 ACM Symposium on Principles of
  Distributed Computing}, pages 217--226. ACM, 2016.

\bibitem[CRT05]{CRT05:MST}
Bernard Chazelle, Ronitt Rubinfeld, and Luca Trevisan.
\newblock Approximating the minimum spanning tree weight in sublinear time.
\newblock {\em SIAM Journal on computing}, 34(6):1370--1379, 2005.

\bibitem[CS09]{CS09:MST}
Artur Czumaj and Christian Sohler.
\newblock Estimating the weight of metric minimum spanning trees in sublinear
  time.
\newblock {\em SIAM Journal on Computing}, 39(3):904--922, 2009.

\bibitem[DHZ19]{duan2019dynamic}
Ran Duan, Haoqing He, and Tianyi Zhang.
\newblock Dynamic edge coloring with improved approximation.
\newblock In {\em Proceedings of the Thirtieth Annual ACM-SIAM Symposium on
  Discrete Algorithms}, pages 1937--1945. SIAM, 2019.

\bibitem[DKM{\etalchar{+}}94]{DKMHRT94:hashing}
Martin Dietzfelbinger, Anna~R. Karlin, Kurt Mehlhorn, Friedhelm {Meyer auf der
  Heide}, Hans Rohnert, and Robert~Endre Tarjan.
\newblock Dynamic perfect hashing: Upper and lower bounds.
\newblock {\em {SIAM} J. Comput.}, 23(4):738--761, 1994.

\bibitem[EGIN97]{EppsteinGIN97}
David Eppstein, Zvi Galil, Giuseppe~F. Italiano, and Amnon Nissenzweig.
\newblock {Sparsification - a technique for speeding up dynamic graph
  algorithms}.
\newblock {\em J. ACM}, 44(5):669--696, 1997.

\bibitem[Fre83]{Fred83}
Greg~N. Frederickson.
\newblock Data structures for on-line updating of minimum spanning trees
  (preliminary version).
\newblock In {\em STOC}, pages 252--257, 1983.

\bibitem[GP13]{GuptaP13}
Manoj Gupta and Richard Peng.
\newblock Fully dynamic $(1+\epsilon)$-approximate matchings.
\newblock In {\em Foundations of Computer Science (FOCS), 2013 IEEE 54th Annual
  Symposium on}, pages 548--557. IEEE, 2013.

\bibitem[HDLT01]{HLT01:connectivity}
Jacob Holm, Kristian De~Lichtenberg, and Mikkel Thorup.
\newblock Poly-logarithmic deterministic fully-dynamic algorithms for
  connectivity, minimum spanning tree, 2-edge, and biconnectivity.
\newblock {\em Journal of the ACM (JACM)}, 48(4):723--760, 2001.

\bibitem[HK97]{HK97:mst}
Monika~R Henzinger and Valerie King.
\newblock Maintaining minimum spanning trees in dynamic graphs.
\newblock In {\em International Colloquium on Automata, Languages, and
  Programming}, pages 594--604. Springer, 1997.

\bibitem[HK99]{HenzingerK99}
Monika~Rauch Henzinger and Valerie King.
\newblock Randomized fully dynamic graph algorithms with polylogarithmic time
  per operation.
\newblock {\em J. ACM}, 46(4):502--516, 1999.

\bibitem[HK01]{HenzingerK01}
Monika~Rauch Henzinger and Valerie King.
\newblock Maintaining minimum spanning forests in dynamic graphs.
\newblock {\em {SIAM} J. Comput.}, 31(2):364--374, 2001.

\bibitem[HKNO09]{HKNO:local}
Avinatan Hassidim, Jonathan~A Kelner, Huy~N Nguyen, and Krzysztof Onak.
\newblock Local graph partitions for approximation and testing.
\newblock In {\em Foundations of Computer Science, 2009. FOCS'09. 50th Annual
  IEEE Symposium on}, pages 22--31. IEEE, 2009.

\bibitem[HRWN15]{HRW15:MST}
Jacob Holm, Eva Rotenberg, and Christian Wulff-Nilsen.
\newblock Faster fully-dynamic minimum spanning forest.
\newblock In {\em Algorithms-ESA 2015}, pages 742--753. Springer, 2015.

\bibitem[KKM13]{KapronKM13}
Bruce~M. Kapron, Valerie King, and Ben Mountjoy.
\newblock Dynamic graph connectivity in polylogarithmic worst case time.
\newblock In {\em SODA}, pages 1131--1141, 2013.

\bibitem[MMPS17]{MMPS17:testable}
Morteza Monemizadeh, S~Muthukrishnan, Pan Peng, and Christian Sohler.
\newblock Testable bounded degree graph properties are random order streamable.
\newblock In {\em 44th International Colloquium on Automata, Languages, and
  Programming (ICALP 2017)}, 2017.

\bibitem[MR09]{MR09:distance}
Sharon Marko and Dana Ron.
\newblock Approximating the distance to properties in bounded-degree and
  general sparse graphs.
\newblock {\em ACM Transactions on Algorithms (TALG)}, 5(2):22, 2009.

\bibitem[MU05]{MU05:probability}
Michael Mitzenmacher and Eli Upfal.
\newblock {\em Probability and computing: Randomized algorithms and
  probabilistic analysis}.
\newblock Cambridge university press, 2005.

\bibitem[NO08]{NO08:constant}
Huy~N Nguyen and Krzysztof Onak.
\newblock Constant-time approximation algorithms via local improvements.
\newblock In {\em Foundations of Computer Science, 2008. FOCS'08. IEEE 49th
  Annual IEEE Symposium on}, pages 327--336. IEEE, 2008.

\bibitem[NS17]{NS17dynamic}
Danupon Nanongkai and Thatchaphol Saranurak.
\newblock Dynamic spanning forest with worst-case update time: adaptive, {Las
  Vegas}, and {$O(n^{1/2-\varepsilon})$}-time.
\newblock In {\em Proceedings of the 49th Annual ACM SIGACT Symposium on Theory
  of Computing}, pages 1122--1129. ACM, 2017.

\bibitem[NSWN17]{NSW17:MST}
Danupon Nanongkai, Thatchaphol Saranurak, and Christian Wulff-Nilsen.
\newblock Dynamic minimum spanning forest with subpolynomial worst-case update
  time.
\newblock In {\em Foundations of Computer Science (FOCS), 2017 IEEE 58th Annual
  Symposium on}, pages 950--961. IEEE, 2017.

\bibitem[ORRR12]{ORRR12:VC}
Krzysztof Onak, Dana Ron, Michal Rosen, and Ronitt Rubinfeld.
\newblock A near-optimal sublinear-time algorithm for approximating the minimum
  vertex cover size.
\newblock In {\em Proceedings of the twenty-third annual ACM-SIAM symposium on
  Discrete Algorithms}, pages 1123--1131. SIAM, 2012.

\bibitem[PD06]{PatrascuD06}
Mihai Patrascu and Erik~D. Demaine.
\newblock Logarithmic lower bounds in the cell-probe model.
\newblock {\em {SIAM} J. Comput.}, 35(4):932--963, 2006.

\bibitem[PR07]{PR07:sublinear_distributed}
Michal Parnas and Dana Ron.
\newblock Approximating the minimum vertex cover in sublinear time and a
  connection to distributed algorithms.
\newblock {\em Theoretical Computer Science}, 381(1):183--196, 2007.

\bibitem[PS18]{PS18:stream}
Pan Peng and Christian Sohler.
\newblock Estimating graph parameters from random order streams.
\newblock In {\em Proceedings of the Twenty-Ninth Annual ACM-SIAM Symposium on
  Discrete Algorithms}, pages 2449--2466. SIAM, 2018.

\bibitem[Sol16]{Sol16:matching}
Shay Solomon.
\newblock Fully dynamic maximal matching in constant update time.
\newblock In {\em Foundations of Computer Science (FOCS), 2016 IEEE 57th Annual
  Symposium on}, pages 325--334. IEEE, 2016.

\bibitem[Sol18]{solomon2018local}
Shay Solomon.
\newblock Local algorithms for bounded degree sparsifiers in sparse graphs.
\newblock In {\em 9th Innovations in Theoretical Computer Science Conference
  (ITCS 2018)}. Schloss Dagstuhl-Leibniz-Zentrum fuer Informatik, 2018.

\bibitem[SW18]{SW18:coloring}
Shay Solomon and Nicole Wein.
\newblock Improved dynamic graph coloring.
\newblock In {\em 26th Annual European Symposium on Algorithms}, 2018.

\bibitem[WN17]{wulff2017fully}
Christian Wulff-Nilsen.
\newblock Fully-dynamic minimum spanning forest with improved worst-case update
  time.
\newblock In {\em Proceedings of the 49th Annual ACM SIGACT Symposium on Theory
  of Computing}, pages 1130--1143. ACM, 2017.

\bibitem[YYI12]{YYI12:constant}
Yuichi Yoshida, Masaki Yamamoto, and Hiro Ito.
\newblock Improved constant-time approximation algorithms for maximum matchings
  and other optimization problems.
\newblock {\em SIAM Journal on Computing}, 41(4):1074--1093, 2012.

\end{thebibliography}

\newpage
\appendix
\begin{center}\huge\bf Appendix \end{center}

\section{Lower Bound for Dynamic $\Delta$-Colorability Testing}\label{sec:lb_delta_coloring}
In~\cite{PatrascuD06} Patrascu and Demaine construct an $n$-node graph and show that there exists a 
sequence $\mathcal S$ of $T$ edge insertion, edge deletion, and query operations such that any data structure for dynamic connectivity must perform  
$\Omega(T \log n)$ cell probes to process the sequence, where each cell has size $O(\log n)$. This 
shows that the amortized number of cell probes per operation is $\Omega(\log n)$.

We now show  how to use this result to get a lower bound for the following {\em dynamic $\Delta$-colorability testing problem}: An insert($u,v)$ operation inserts the edge $(u,v)$, a delete($u,v$) operation deletes the edge
$(u,v)$, and a query($u,v$) operation returns {\em yes} if the graph is $\Delta$-colorable and {\em no}
otherwise, where $\Delta$ is the maximum degree in the current graph. Specifically we show the lower bound for $\Delta = 2$.

The graph $G$ in the proof of\cite{PatrascuD06} consists of a $\sqrt n \times \sqrt n$ grid, where each node in column $1$ has exactly 1
edge to a node of column 2 and no other edges, each node in column $i$, with $1 < i < \sqrt n$ has exactly 1 edge to a node of column $i-1$ and 1 edge to a node of column $i+1$ and no other edges, and
each node in column $\sqrt n$ has exactly 1 edge to a node of column $\sqrt n -1$ and no other edges.
Thus, the graph consists of $\sqrt n$ paths of length $\sqrt n -1$ and the edges between column $i$ and $i+1$ for any $1 \le i < \sqrt n$ represent a permutation of the $\sqrt n$ rows.
The sequence $\mathcal S$ consists of ``batches'' of $O(\sqrt n)$ edge updates, replacing the permutation of some column $i$ by a new permutation for column $i$. Between the batches of updates are ``batches'' of connectivity queries, each consisting of $\sqrt n$ connectivity queries
and a parameter $1 \le k \le \sqrt{n}$, where the $j$-th query 
for $1 \le j \le \sqrt n$ of each batch tests whether the $j$-th vertex of column 1 is connected with a specific vertex of column $k$. 

Note that the maximum degree $\Delta$ is 2. We now show how to modify each connectivity
query $(u,v)$ such that it consists of a constant number of edge updates and one query whether the resulting graph is $\Delta$-colorable. The answer will be {\em no} iff $u$ and $v$ are connected.
Thus, in the resulting sequence $\mathcal S'$ the number of query operations equals the number of
query operations in $\mathcal S$ and the number of update operations is linear in the number of update and query operations in $\mathcal S$. 
Thus the total number of oerations in $\mathcal S'$ is only a constant factor larger
than the number of operations in $\mathcal S$, which, together with the result of~\cite{PatrascuD06}, implies that the amortized number 
of cell probes per operation is $\Omega(\log n)$.

We now show how to simulate a connectivity query($u,v$), where $u$ is in column 1 and $v$ is in column $k$ for some $1 \le k \sqrt{n}$. 
We assume that $k$ is even and explain below how to deal with the case that $k$ is odd.
The instance for the dynamic $\Delta$-colorability testing consists of $G$ with an additional node $s$ added.
To simulate a connectivity query($u,v$) we (1) remove the edge from $v$ to
its neighbor in column $k+1$ if $k < \sqrt{n}$, (2) add the edges $(u,s)$ and $(v,s)$ and then (3) ask a $\Delta$-colorability query. Note that the resulting graph still has maximum degree 2.
Furthermore, if $u$ and $v$ are connected in $G$ then there exists a unique path of odd length $k -1$ between them. 
Together with the edges $(u,s)$ and $(v,s)$ and the assumption that $k$ is even, this results in an odd length cycle, so that the answer to the  2-colorability query is {\em no}. 
If, however, $u$ and $v$ are not connected in $G$, then adding the edges $(u,s)$ and $(v,s)$ creates a path of length $2 + \sqrt n - 1 + k - 1 = \sqrt n + k$, but no cycle. Thus, the 2-colorability query returns {\em yes}.
Thus $u$ and $v$ are connected in $G$ iff the 2-colorability query in the modified graph returns {\em no}.
Afterwards we remove the edges $(u,s)$ and $(v,s)$.
Finally if $k$ is odd, we do not add a vertex $s$ to $G$ and to simulate the connectivity query($u,v$)
we simply insert the edge $(u,v)$. As before there exists an odd length cycle in the graph iff $u$ and $v$ are connected. The rest of the proof remains unchanged.

This shows the following theorem.
\begin{theorem}
Any data structure for dynamic $\Delta$-colorability testing, where $\Delta$ is the maximum degree in the graph, must perform $\Omega(\log n)$ cell probes, where each cell has size $O(\log n)$.
\end{theorem}

\section{Deferred Implementations, Proofs and Discussions from Section~\ref{sec:coloring}}\label{sec:appendix_sec_coloring}
\subsection{Updating the Data Structures}\label{sec:appendix_data_structure}
\paragraph{Case I: an edge deletion $(u,v)$.} Whenever an edge $(u,v)$ gets deleted, we update the data structures corresponding to $u$ and $v$ as follows. More precisely, we first update the sets $H_u, L_u, H_v, L_v$ and their sizes trivially in constant time. 
The lists $\ColorH{u},\ColorL{u},\ColorH{v}, \ColorL{v}$ can be updated in constant worst-case time. The hash tables $\mathcal{A}_u^H, \mathcal{A}_v^H$ can also be maintained in constant amortized expected update time. %
More precisely, suppose w.l.o.g., $u\in L_{v}$, then we do the following:
\begin{enumerate}
	\item Delete $(\chi(v),\mu_u^H(\chi(v)))$ from $\mathcal{A}_u^H$; $\mu_u^H(\chi(v))\gets \mu_u^H(\chi(v)) -1$. 
	\item If $\mu_u^H(\chi(v))=0$, then $\ColorH{u} \gets \ColorH{u}\setminus \{\chi(v)\}$, $\ColorL{u} \gets\ColorL{u} \cup \{\chi(v)\} $.
	\item Otherwise, insert $(\chi(v),\mu_u^H(\chi(v)))$ to $\mathcal{A}_u^H$.
\end{enumerate}
\paragraph{Case II: an edge insertion $(u,v)$ such that $\chi(u)\neq \chi(v)$.} In this case, w.l.o.g., suppose that $r(u)<r(v)$, we update the data structures as follows: 
\begin{enumerate}
	\item $\ColorH{u} \gets \ColorH{u}\cup\{\chi(v) \}$, 
	$\ColorL{u}\gets \ColorL{u}\setminus \{\chi(v)\}$, $\mu_u^H(\chi(v))\gets \mu_u^H(\chi(v)) + 1$
	\item Delete $(\chi(v),\mu_u^H(\chi(v))-1)$ from $\mathcal{A}_u^H$ if $\mu_u^H(\chi(v))>1$, insert $(\chi(v),\mu_u^H(\chi(v)))$ to $\mathcal{A}_u^H$.
\end{enumerate}

\paragraph{Case III: procedure $(\divideontimes)$ in the subroutine \textsc{Recolor}($v$).} In the subroutine \textsc{Recolor($v$)}, if the color of $v$ is changed from $c'$ to $c$, then we update the relevant data structure as follows: 

\begin{itemize}%
\item[$(\divideontimes)$] {For every $w\in L_v$: \label{procedure}}
\begin{enumerate}
	\item $\mu_w^H(c')\gets \mu_w^H(c') - 1$
	\item If $\mu_w(c')=0$, then $\ColorH{w}\gets \ColorH{w} \setminus\{c' \}$, $\ColorL{w}\gets \ColorL{w}\cup \{c'\}$, 
	\item $\ColorH{w}\gets \ColorH{w}\cup\{c \}$, $\ColorL{w}\gets \ColorL{w}\setminus\{c\}$, $\mu_w^H(c)\gets \mu_w^H(c)+1$.
	\item Delete $(c,\mu_w^H(c))$ from $\mathcal{A}_w^H$ if $\mu_w^H(c)>1$, and insert $(c,\mu_w^H(c))$ to $\mathcal{A}_w^H$.
\end{enumerate}

\end{itemize}

\subsection{Initialization in $O(n)$ Time}\label{subsec:initialization}

Now we describe how we can reduce the initialization time from $O(n \Delta)$ to $O(n)$. Note that the only part that takes $O(n\Delta)$ time is to initialize $\ColorL{u}$ for each vertex $u$, and the rest part of initialization already only takes $O(n)$ time. 
The main observation is that $\ColorL{u}$ is only needed in the sampling subroutine of \textsc{SetColor}($u)$ and even there only once the degree of
a vertex is at least $\Delta/2$. Since we make the standard assumption that we start with an empty graph, this means that $\Omega(\Delta)$ insertions incident to $u$ must have happened. Thus, we build $\ColorL{u}$ only once this is the case and amortize the cost of building it over these previous $\Omega(\Delta)$ insertions.

To be more precise, we change  the initialization phase as follows:
We do not build $\ColorL{u}$ for any vertex $u$. Note that all other data structure are built as before, but they only have size $O(n)$ and only take time $O(n)$ to build.

When an edge $(u,v)$ is inserted, we check whether one of the endpoints, say $u$, of the newly inserted edge reaches the degree $\Delta/2$ and does not yet have the data structure $\ColorL{u}$. If so, we build $\ColorL{u}$ and its hash table at this point in time $O(\Delta)$. We amortize this cost over the $\Delta/2$ updates that increased the degree of $u$ to $\Delta/2$, adding a constant amortized cost to each of them. (If the other endpoint $v$ also reaches the degree $\Delta/2$, we handle it analogously.)

Note that this does not affect the \textsc{SetColor} algorithm: as long as the degree of a vertex $u$ is less than $\Delta/2$, \textsc{SetColor}($u$) selects a new color by sampling in Step~\ref{alg:recolor_low_degree_2} from 
$\mathcal{B}_u$. To do so $\ColorL{u}$ is not needed:
In time $O(|L_u|)$ time we build the lists and corresponding hash tables for $\ColorM{u}{L} \cup \ColorU{u}{L}$, which together with the maintained list and hash table for $\ColorH{u}$ suffice for us to sample a color from $\mathcal{B}_u$ in $O(1)$ time: We pick a random color from $\cal C$ and test whether it belongs to $\mathcal{B}_u$ by making sure that it does not
belong to $\ColorM{u}{L} \cup \ColorU{u}{L}$ or $\ColorH{u}$. The fact that the degree of $u$ is at most $\Delta/2$ implies that in expectation the second randomly chosen color will belong to $\mathcal{B}_u$.

Once $\ColorL{u}$ and its hash table has been built, it is used in the way as we described before and updated as in Section~\ref{sec:appendix_data_structure}.

\subsection{Deferred Proofs}\label{subsec:app_coloring_proof}

\paragraph{Proof of Lemma \ref{lemma:recolor}.} In the following, we provide the proof of Lemma \ref{lemma:recolor}.
\begin{proof}%
	Note that the first item of the lemma follows from Step \ref{alg:recolor_low_degree_2} of the subroutine \textsc{Recolor}($v$), as the algorithm samples a new color from $\mathcal{B}_v$, whose size is larger than $\Delta+1-\frac{\Delta}{2}= \frac{\Delta}{2}+1$.
	
	Now we consider the second item, i.e., the case that  $|L_v|+|H_v|\geq \frac{\Delta}{2}$. %
	Recall that $\ColorM{v}{L}$ denote the set of colors that have been used by at least two vertices in $L_v$. We first note that
	$$|\ColorM{v}{L}|+ |\mathcal{B}_v| + |\ColorU{v}{L}|=|\ColorM{v}{L}\cup \mathcal{B}_v\cup\ColorU{v}{L}|\geq |\mathcal{C}|-|H_v|= \Delta+1 - |H_v|\geq 1+|L_v|.$$
	Furthermore, by definition of $\ColorM{v}{L}$, it holds that $2|\ColorM{v}{L}|+|\ColorU{v}{L}|\leq |L_v|$. Thus, 
	\begin{eqnarray}
	2|\mathcal{B}_v|+|\ColorU{v}{L}|\geq 2+2|L_v| - |L_v|\geq 2+|L_v|\label{eqn:bvanduv}
	\end{eqnarray}
	
	Now we distinguish two cases. If Step \ref{alg:manynew} happens, i.e., $|\NewL{v}|\geq \frac{1}{10}|L_v|$, then $|\GoodNewL{v}|\geq \frac{1}{20}|L_v|$. This further implies that $|L_v\setminus \GoodNewL{v}|\leq \frac{19}{20}|L_v|$ and thus $|\ColorU{v}{L}\setminus \ColorU{v}{L_{\textrm{new}}^{\textrm{g}}}|\leq |L_v\setminus \GoodNewL{v}|\leq \frac{19}{20}|L_v|$.
	Then if $|\ColorB{v}|\geq \frac{1}{50}|L_v|+1$, then $s=\min\{|\mathcal{B}_v\cup \ColorU{v}{L_{\textrm{new}}^{\textrm{g}}}|, |\GoodNewL{v}|+1\}\geq \frac{1}{50}|L_v|+1$; otherwise (i.e., $|\ColorB{v}|\leq \frac{1}{50}|L_v|)$, by inequality (\ref{eqn:bvanduv}), $|\ColorU{v}{L}|\geq 2 + |L_v| -\frac{2}{50} |L_v|$. Thus
	\begin{eqnarray*}
		|\ColorU{v}{L_{\textrm{new}}^{\textrm{g}}}| = |\ColorU{v}{L}| - |\ColorU{v}{L}\setminus \ColorU{v}{L_{\textrm{new}}^{\textrm{g}}}| \geq 2 + |L_v| -\frac{2}{50} |L_v| - \frac{19}{20}|L_v| = \frac{1}{100}|L_v|+2
	\end{eqnarray*}
	This gives that $s=\min\{|\mathcal{B}_v\cup \ColorU{v}{L_{\textrm{new}}^{\textrm{g}}}|, |\GoodNewL{v}|+1\}\geq \frac{1}{100}|L_v|+1$.

	If Step \ref{alg:manyold} happens, i.e., $|\OldL{v}|> \frac{9}{10}|L_v|$, then $|\GoodOldL{v}|> \frac{9}{20}|L_v|$. Thus $|L_v\setminus \GoodOldL{v}|< \frac{1}{20}|L_v|$ and $|\ColorU{v}{L}\setminus \ColorU{v}{L_{\textrm{old}}^{\textrm{g}}}|\leq |L_v\setminus \GoodOldL{v}|< \frac{1}{20}|L_v|$. Furthermore, 
	\begin{eqnarray*}
		|\mathcal{B}_v\cup \ColorU{v}{L_{\textrm{old}}^{\textrm{g}}}| = |\mathcal{B}_v| + |\ColorU{v}{L_{\textrm{old}}^{\textrm{g}}}| %
		&=&|\mathcal{B}_v| + |\ColorU{v}{L}|-|\ColorU{v}{L}\setminus \ColorU{v}{L_{\textrm{old}}^{\textrm{g}}}|\\
		&\geq& |\mathcal{B}_v| + \frac{|\ColorU{v}{L}|}{2}-|\ColorU{v}{L}\setminus \ColorU{v}{L_{\textrm{old}}^{\textrm{g}}}|\\
		&\geq& 1+ \frac{|L_v|}{2} - \frac{1}{20}|L_v| >1+\frac{1}{20}|L_v|,
	\end{eqnarray*}
	where in the last inequality we use the inequality (\ref{eqn:bvanduv}). Thus  $s=\min\{|\mathcal{B}_v\cup \ColorU{v}{L_{\textrm{old}}^{\textrm{g}}}|, |\GoodOldL{v}|+1\}\geq 1+\frac{1}{20}|L_v|$. This finishes the proof of the Lemma.	
\end{proof}

\subsection{Extension to Work for Changing $\Delta$}\label{sec:chaningdelta}
As we mentioned, we can extend our algorithm to work with changing $\Delta$. (A similar extension was also done in~\cite{BCHN18:coloring}). 
For any time stamp $t\geq 0$, we will maintain a global value $\Delta_t:= \max_{j=1}^t\max_{v\in V} \deg_j(v)$, where $\deg_j(v)$ denotes the degree of $v$ in the graph after $j$ edge updates, that is, $\Delta$ is the maximum degree seen so far (till time $t$). 
Then we have a randomized algorithm for maintaining a $(\Delta_{t}+1)$-coloring. 
More precisely, for any time stamp $j$, for each vertex $v$, we only need to guarantee that the color $\chi(v)$ is chosen from $\{1, \dots,\deg_j(v)+1\}$. 
Then for each vertex $v\in V$, we let $\ColorL{v}\subseteq \mathcal{C}$ consist of all the colors in $\{1, \dots, \deg_j(v)+1 \}$ that have not been assigned to any neighbor $u$ of $v$ for $u\in H_v$. 
It is easy to see that Lemma~\ref{lemma:setcolor}, \ref{lemma:recolor_time} and \ref{lemma:recolor} still hold, and our randomized dynamic coloring algorithm maintains a proper $(\Delta_t+1)$-coloring of the graph $G_t$ at time $t$ with constant amortized update time, for any $t\geq 0$.

Additionally we can keep a variable $\Delta$ such that we rebuild the
data structure every $\Delta n$ operations as follows:
We determine the list of current edges and set $\Delta$ to be the 
maximum degree of the current graph. Then we build the data structure for an empty graph and insert all edges using the insert operation. This increases 
the running time by an amortized constant factor and guarantees that $\Delta$ is the maximum degree in the graph {\em within the last $\Delta n$ updates.}

\section{A Note on Dynamically Estimating the Number of CCs}\label{app:4}
\paragraph{Estimating $\ncc(G)$ with an additive error $\varepsilon n^{O(1)}$.} We note that similar to our previous algorithm (in Section~\ref{subsec:proof_random_ncc}) for estimating $\ncc(G)$ with an additive error $\varepsilon \Thr(G)$, if we simply invoke the static algorithm from Lemma~\ref{lemma:bkm_ncc} to obtain an estimator $\overline{cc}$ for the current graph and re-compute the estimator every $\Theta(\varepsilon n)$ updates, then the corresponding algorithm always maintain an estimator for $\ncc(G)$ with an additive error $\varepsilon n$. That is, we have the following theorem.
\begin{theorem}\label{thm:random_ncc_espn}
	Let $1>\varepsilon >0$ and $0<p<1$. There exists a fully dynamic algorithm that with probability at least $1-p$, maintains an estimator $\overline{cc}$ for the number $\ncc$ of CCs of a graph $G$ s.t.,  $|\overline{cc}-\ncc(G)|\leq \varepsilon\cdot n$. The worst-case time per update operation is $O(\max\{1,\frac{\log(1/\varepsilon)\log(1/p)}{\varepsilon^3 n}\})$.
\end{theorem}

The following is a direct corollary of the above theorem.
\begin{corollary}\label{cor:random_NCC}
	Let $\varepsilon >0$ and let $c$ be any constant such that $c\geq 1$. There exists a fully dynamic algorithm that with probability at least $1-\frac{1}{n^c}$, maintains an estimator $\overline{cc}$ for the number $\ncc$ of CCs of a graph $G$ s.t.,  $|\overline{cc}-\ncc(G)|\leq \varepsilon n^{2/3}\log^{2/3}n$. The worst-case time per update operation is $O(\varepsilon^{-3})$. 
\end{corollary}
\begin{proof}
	If $\varepsilon < n^{-\frac23}$, then for each update, one can use the naive BFS algorithm to exactly compute $\ncc(G)$, which runs in time $O(n^2)=O(\varepsilon^{-3})$. If $\varepsilon \geq n^{-\frac23}$, we can apply Theorem~\ref{thm:random_ncc_espn} with parameters $p=\frac{1}{n^c}$, $\varepsilon' = \varepsilon n^{-1/3}\cdot\log^{2/3}n$, to obtain an $\overline{cc}$ for $\ncc(G)$ with an additive error $\varepsilon n^{2/3}\log^{2/3}n$. The corresponding dynamic algorithm has update time $O(\varepsilon^{-3}\cdot (1+\frac{\log(1/\varepsilon)}{\log n}))=O(\varepsilon^{-3})$.%
\end{proof}

\textbf{Remark:} We cannot expect to be able to get a constant-time algorithm for maintaining the number of connected components with an additive error of $1$ or a multiplicate error of $2$: Any such algorithm would be able to decide whether the graph is connected or not, contradicting the  $\Omega(\log n)$ lower bound for dynamically maintaining whether a graph is connected~\cite{PatrascuD06}.

\end{document}